\documentclass[journal]{IEEEtran}
\providecommand{\CleanCopy}{1}

\usepackage{amsmath,amssymb,mathtools}
\usepackage{amsthm}
\usepackage{graphicx}
\usepackage[caption=false,font=footnotesize]{subfig}
\usepackage{booktabs}
\usepackage{needspace}
\usepackage{float}
\usepackage{algorithm}
\usepackage{algorithmic}
\usepackage{cite}
\usepackage{tikz}
\usetikzlibrary{arrows.meta,positioning,calc,shapes.geometric,shapes.misc}

\usepackage[hidelinks]{hyperref} %

\usepackage{xcolor}
\ifnum\CleanCopy=0\relax
  \colorlet{revisiontext}{blue}
\else
  \colorlet{revisiontext}{black}
\fi
\DeclareRobustCommand{\minorchange}[1]{\textcolor{revisiontext}{#1}}

\newtheorem{lemma}{Lemma}

	\title{Analog RF Computing: A New Paradigm for Energy-Efficient Edge AI Over MU-MIMO Systems}
\author{Wentao Yu, \textit{Member, IEEE}, and Vincent W.S. Wong, \textit{Fellow, IEEE}

\thanks{Wentao Yu and Vincent W.S. Wong
are with the Department of Electrical and Computer Engineering, The University of British Columbia, Vancouver, BC, V6T 1Z4, Canada (e-mail: wentaoyu@ece.ubc.ca; vincentw@ece.ubc.ca). 	}
}

\begin{document}
\bstctlcite{BSTcontrol}
\maketitle

\vspace{-1.2em}
\begin{abstract}

Modern edge devices increasingly rely on neural networks for intelligent applications. However, conventional digital computing-based edge inference requires substantial memory and energy consumption. In particular, matrix-vector multiplications (MVMs) account for most of the computation burden. Recent work has shown that analog radio frequency (RF) computing can reduce the energy consumption. In analog RF computing, a base station (BS) encodes the weights of the neural networks and broadcasts the RF waveforms to the clients. Each client reuses its passive mixer to multiply the received weight-encoded waveform with a locally generated input-encoded waveform. This enables wireless receivers to perform MVMs with ultra-low energy consumption. Unlike conventional downlink transmissions which are optimized for communications, analog RF computing requires a computing-centric physical layer that controls both the analog MVM accuracy and the energy consumption for inference. Motivated by this, in this paper, we propose a physical layer design framework for analog RF computing in multi-user multiple-input multiple-output (MU-MIMO) wireless systems. We derive tractable models for computing accuracy and energy consumption for inference, formulate a joint BS beamforming and client-side scaling problem subject to computing accuracy, transmit power, and hardware constraints, and develop a low-complexity algorithm to solve the nonconvex problem. The proposed design provides client- and layer-specific accuracy control for both uniform- and mixed-precision inference. Simulations under Third Generation Partnership Project (3GPP) specifications show that analog RF computing can significantly reduce client-side energy consumption by nearly two orders of magnitude compared with digital computing, while mixed-precision inference requires even lower energy consumption than uniform-precision inference. Overall, these results establish analog RF computing over wireless networks as a promising paradigm for energy-efficient edge inference. 

\end{abstract}

\section{Introduction}\label{sec:introduction}

The growing demands for artificial intelligence (AI) services on edge devices (e.g., cameras, wearables, robots, and Internet of things (IoT) devices) are pushing neural network (NN) inference towards the network edge \cite{letaief2022edge}. However, some of the existing edge inference paradigms may not be efficient for power-constrained and memory-limited devices. One approach is on-device digital computing, which runs AI inference locally on digital hardware, such as graphics processing units (GPUs) or specialized AI accelerator chips. However, this requires storing the NN model on the device and repeatedly transferring the weights and activations between memory and computing units during inference, which incurs a large memory overhead and high energy consumption \cite{horowitz2014energy}. Another approach is to offload the inference task to edge or cloud servers, which reduces local model storage but introduces latency and privacy concerns \cite{Tang2022deep,Shah-Mansouri2017joint}. These limitations motivate a new inference architecture that can preserve on-device AI inference while reducing local model storage and the heavy memory and energy costs of conventional digital computing. In this paper, we study one such architecture, namely \textit{analog radio frequency (RF) computing}, which repurposes wireless receivers as low-power AI inference engines to overcome these limitations. 

Analog RF computing operates as follows. A base station (BS) broadcasts RF waveforms that encode NN weights to multiple clients over the wireless downlink. Each client obtains its local inference inputs and generates the corresponding input-encoded waveform. It then reuses the passive RF mixer \cite{razavi2011rf}, originally used for in-phase and quadrature (I/Q) demodulation in the RF front end of a wireless receiver, as an analog multiplier to multiply the received weight-encoded waveform with the locally generated input-encoded waveform. In this way, the wireless receiver can directly perform matrix-vector multiplications (MVMs) between NN weights and inference inputs via analog RF computing. The resulting product waveform contains the desired MVM outputs and can be recovered through filtering and readout. Recent prototypes have demonstrated the remarkable inference energy efficiency of this architecture in a point-to-point link \cite{wise}. 

Analog RF computing offers several advantages over conventional edge inference architectures. First of all, it does not require dedicated local digital computing hardware. Any edge device equipped with a wireless receiver can perform on-device inference. Moreover, it reduces local model storage and memory requirements. The model weights are broadcast by the BS. The primary computing workloads of AI inference, i.e., MVMs, are carried out directly in the RF front end without storing weights in memory. Furthermore, the energy consumption required for client-side inference is extremely low. Since the considered RF mixer is a passive device, the primary client-side energy consumption is due to input-waveform generation and readout, which is significantly lower than that of digital computing. Collectively, these features make analog RF computing especially appealing for AI inference on power-constrained edge devices. 

Despite these advantages, analog RF computing introduces new challenges to the physical layer design in wireless systems. In conventional wireless systems, transceivers are optimized mainly for communication-oriented metrics such as data rate, outage probability, and signal-to-noise ratio (SNR) \cite{Bruno2024multiple}. On the other hand, since the wireless downlink is now utilized for edge inference, analog RF computing shifts the design focus to computing-centric metrics, most importantly the computing accuracy at the clients and the energy consumption for inference at both the clients and BS. These metrics depend jointly on the beamformers at the BS side, the scaling coefficients at the client side, the physical layer characteristics of the downlink wireless channels, and the hardware characteristics of passive mixers. Consequently, the physical layer in wireless systems must be redesigned from a computing-centric perspective, which is fundamentally different from existing formulations. 

To the best of our knowledge, this is the first study to formulate and solve the physical layer design problem for analog RF computing in multi-user multiple-input multiple-output (MU-MIMO) wireless systems. The central question is how a BS can utilize the downlink to not only transmit the NN model weights, but also to control the computing accuracy and inference energy consumption at the clients. Solutions to this question will open the door to multiple potential new wireless services, in which one BS can simultaneously serve multiple heterogeneous inference requests from different resource-constrained edge devices while satisfying the client-specific computing accuracy requirements and energy budgets. To this end, we develop tractable mathematical models for analog RF computing in MU-MIMO systems and jointly optimize the BS-side beamformers and the client-side scaling coefficients for minimizing the energy consumption for inference subject to computing accuracy, transmit power, and hardware-specific constraints. We then evaluate the proposed framework under Third Generation Partnership Project (3GPP) channel models and specifications, and provide system-level insights into the deployment of analog RF computing. 

\subsection{Related Works}

Existing literature relevant to this paper can be grouped according to how analog computing is physically realized. The first category studies analog computing-based AI inference beyond wireless systems. Optical and photonic processors, analog in-memory computing, and metamaterial-based computing all use physical dynamics to perform linear operations more efficiently than digital computing platforms \cite{shen2017nanophotonic,ielmini2018imc,silva2014metamaterials,liu2022metasurface}. Their common limitation is that they rely on specialized computing hardware instead of wireless receivers which are widely available in edge devices. 

The second category is wireless analog computation through propagation or superposition. Classical analog function computation and over-the-air computation (AirComp) utilize the multiple-access wireless channel as an analog adder. They are designed for addition or aggregate functions across distributed transmitters \cite{goldenbaum2013analogfunc,yang2020federated,li2023iscco,wen2024taskaircomp}. AirNN \cite{sanchez2022airnn} and AirFC \cite{reusmuns2023airfc,hua2026risnn} use wireless propagation to emulate convolutional and fully connected layers for NN inference where the radio environment and transmitted waveforms are carefully designed. For example, AirNN uses reconfigurable intelligent surfaces to create signal reflections that emulate the desired convolution filters~\cite{sanchez2022airnn}. Microwave linear analog computers (MiLACs) utilize reconfigurable microwave networks to perform analog signal processing and communications tasks such as matrix inversion and channel estimation \cite{milac_part1,Matteo2025TSP2,Matteo2026MIMO}. The aforementioned works are related to ours in spirit, but they rely on wireless propagation or superposition as part of the computing process and hence require the wireless channel or radio environment to be carefully engineered for specific computing tasks. On the other hand, our considered analog RF computing does not use the wireless channel itself as the computing medium. The analog computing takes place solely in the RF front end of wireless receivers. 

The third category is RF computing in IoT systems, which offers an energy-efficient approach to process information using RF-domain signal operations \cite{rf_computing_realm}. RF-Transformer~\cite{rftransformer} provides a programmable backscatter abstraction for synthesizing transmissions compatible with multiple communication protocols. Sense Me on the Ride~\cite{sense_me_ride} uses LoRa backscatter and receiver-side signal processing for sensing. WindScatter~\cite{windscatter} directly couples a tunneling magnetoresistance wind sensor to backscatter modulation for low-power sensing and communication. These works highlight the potential of RF computing for communication and sensing in IoT systems.

The fourth and closest category is RF-domain inference with broadcast NN model weights. The authors of \cite{wise} experimentally demonstrated over-the-air weight broadcasting and passive mixer-based inference on a software-defined radio platform. The use of mixers for analog MVMs was originally proposed in this work. Specifically, the work in \cite{wise} showed the device-level feasibility of using passive mixers to carry out analog MVM under fixed transceiver configurations. However, it did not address how a wireless system should allocate physical layer resources when analog RF computing becomes a downlink edge inference service. This paper fills this gap by developing a computing-centric physical layer design for analog RF computing in MU-MIMO wireless systems. 

\begin{figure*}[t]
  \centering
  \includegraphics[width=0.99\textwidth]{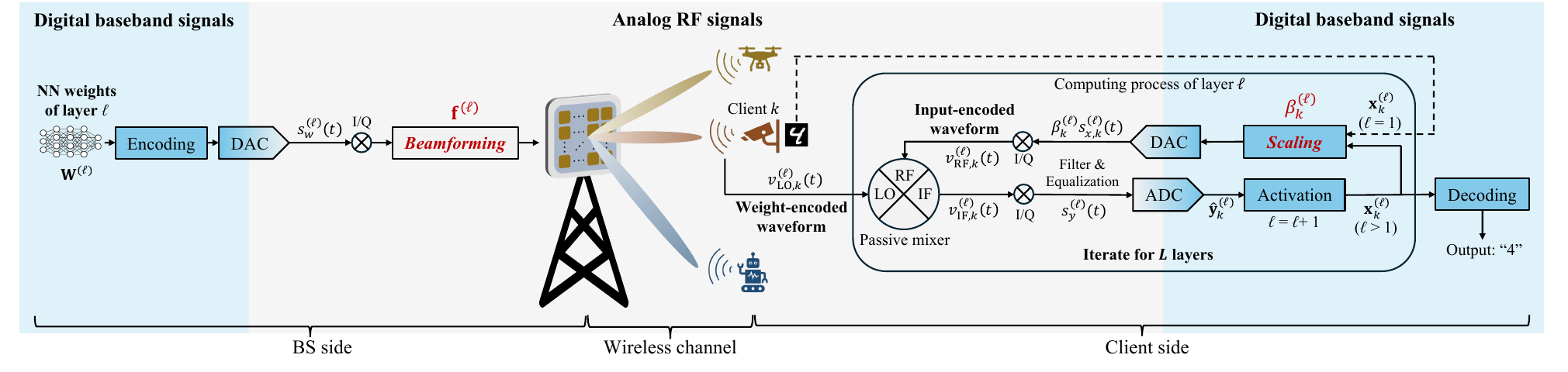}
  \caption{Analog RF computing-based edge inference for representative client $k$. At layer $\ell$, the BS broadcasts a beamformed weight-encoded waveform. Each client multiplies it with a locally generated and scaled input-encoded waveform. The MVM output is read out after filtering, equalization, and ADC of the analog signals. The input to the next layer is obtained by applying the activation function and re-encoding the resulting output. The final inference result is obtained after all NN layers have been processed.}
  \label{fig:system_schematic}
\end{figure*}

\subsection{Contributions}

This paper proposes an analog RF computing-based edge AI inference framework in MU-MIMO systems. Specifically, a multi-antenna BS broadcasts the NN weight-encoded waveforms to multiple clients. Each client reuses the passive mixer available in its wireless receiver to perform analog MVMs with its locally generated input-encoded waveforms. For this novel wireless system, we study the computing-centric physical layer design. The main contributions are summarized as follows.

\begin{itemize}
    \item We establish the baseband and RF waveform construction for analog RF computing in MU-MIMO wireless systems and derive a tractable baseband equivalent model. Building on this, we explicitly characterize how the BS-side beamforming and client-side scaling coefficients jointly determine the post-equalization normalized mean squared error (NMSE) of analog MVMs, as well as the inference energy consumption at both the BS and client sides. 
    
    \item We formulate a joint BS-client energy minimization problem for simultaneously serving multiple clients with heterogeneous computing accuracy requirements. The formulated problem incorporates per-client NMSE constraints, BS-side and client-side power limits, and hardware constraints of the mixers. The objective jointly captures BS-side and client-side inference energy consumption, including waveform generation, analog-to-digital conversion (ADC) \cite{razavi2011rf}, and digital decoding. It reveals a tradeoff between the energy consumption at the BS and the energy consumption at the clients for analog RF computing. 
  
    \item We derive the closed-form optimal client-side scaling coefficients and prove that an optimal BS beamformer always lies in the channel subspace. Based on these, we develop a reduced-dimensional successive convex approximation (SCA) algorithm. Its complexity scales with the number of clients rather than that of BS antennas. Simulation results show that the proposed algorithm matches its full-dimensional counterpart with substantially lower runtime. Results also demonstrate that analog RF computing enables highly energy-efficient edge inference compared with conventional digital computing under realistic settings according to 3GPP specifications. 
    
    \item We further discuss uniform- and mixed-precision inference across NN layers. Uniform-precision inference sets the same computing accuracy target for all layers, whereas mixed-precision inference sets different targets for different layers. Since the per-layer computing accuracy can be controlled by the BS beamformer and client-side scaling, the proposed analog RF computing framework can also optimize layerwise computing accuracy targets under an energy budget. Numerical results show that mixed precision further improves the inference performance under the same client-side energy budget. 
\end{itemize}

\subsection{Paper Organization and Notations}

The remainder of this paper is organized as follows. Section~\ref{sec:mixer_mvm} presents the system model and foundations of analog RF computing-based on-device AI inference. Section~\ref{sec:problem_formulation} derives the tractable models for the computing accuracy and energy consumption of analog RF computing, and formulates the optimization problem for computing-centric physical layer design. Section~\ref{sec:solver} develops the proposed reduced-dimensional SCA algorithm. Section~\ref{sec:uniform_mixed_precision} discusses uniform- and mixed-precision inference. The performance evaluation is presented in Section~\ref{sec:sim_results}. Conclusion is given in Section~\ref{sec:conclusion}.

\textit{Notations:} In this paper, scalars, vectors, and matrices are denoted by italic letters (e.g., $a$), boldface lowercase letters (e.g., $\mathbf{a}$), and boldface uppercase letters (e.g., $\mathbf{A}$), respectively. $(\cdot)^{\mathsf T}$ and $(\cdot)^{\mathsf H}$ refer to transpose and Hermitian transpose, respectively. $(\cdot)^\ast$ denotes complex conjugation, and $\mathbb{C}$ denotes the set of complex numbers. $\Re\{\cdot\}$ returns the real part of a complex quantity, $\|\cdot\|_2$ denotes the Euclidean norm, and $\mathcal{CN}(\boldsymbol{\mu},\mathbf{\Sigma})$ denotes a complex Gaussian distribution with mean $\boldsymbol{\mu}$ and covariance $\mathbf{\Sigma}$. $\operatorname{col}(\mathbf{A})$ denotes the column space of $\mathbf{A}$.

\section{System Model of Analog RF Computing-Based Edge Inference}\label{sec:mixer_mvm}

NN inference consists of repeated linear operations interleaved with lightweight nonlinear operations. The linear operations that account for most of the computation can be represented as MVMs between the weights and inputs of each layer \cite{sze2017efficient}. Consider that the NN has $L$ layers. Let $\mathcal{L}\triangleq\{1,\ldots,L\}$ denote the set of layer indices. In layer $\ell\in\mathcal{L}$, the weight matrix is $\mathbf{W}^{(\ell)}\in\mathbb{C}^{M^{(\ell)}\times N^{(\ell)}}$. Let $\mathcal{K}\triangleq\{1,\ldots,K\}$ denote the set of clients. Each client $k\in\mathcal{K}$ has an input vector $\mathbf{x}_k^{(\ell)}\in\mathbb{C}^{N^{(\ell)}}$. For the first layer (i.e., $\ell=1$), the vector $\mathbf{x}_k^{(1)}$ is the initial input to the network for client $k$, such as raw images captured by an edge camera. For layer $\ell>1$, $\mathbf{x}_k^{(\ell)}$ is the layer activation, i.e., the previous layer's output after the activation function. The desired output of the MVM between the weights and inputs is denoted by $\mathbf{y}_k^{(\ell)}\triangleq \mathbf{W}^{(\ell)}\mathbf{x}_k^{(\ell)}{\in\mathbb{C}^{M^{(\ell)}}}$, after which the activation function and other nonlinear operations produce the input of the next layer $\mathbf{x}_k^{(\ell+1)}$. The MVM output, denoted by $\widehat{\mathbf{y}}_k^{(\ell)}$ in Fig.~\ref{fig:system_schematic}, consists of the desired output $\mathbf{y}_k^{(\ell)}$ plus the thermal noise introduced during analog RF computing. This process is repeated until the MVM operations of all layers are finished.

The MVM for each layer is performed by analog RF computing in the RF front end of wireless receivers, as shown in Fig.~\ref{fig:system_schematic}. The passive mixer is a three-port device. It includes the local oscillator (LO) port, RF port, and intermediate frequency (IF) port \cite{razavi2011rf}. For each NN layer, the BS broadcasts a weight-encoded RF waveform to drive the LO port of the passive mixer at each client. Each client then locally generates an input-encoded RF waveform to drive the RF port. The client aligns the two waveforms in time while maintaining their intended carrier frequency difference. The passive mixer multiplies the two RF waveforms to obtain the desired MVM through analog RF computing \cite{wise}. The MVM results are read out from the IF-port signals and digitized before the lightweight nonlinear operations, such as activation function, are performed at the client. The resulting activations are then mapped back to an RF waveform for the next layer, which is repeated until the last NN layer. Each client reuses this processing chain across layers, with layer-specific scaling coefficients. Lastly, the final layer outputs the inference result.

In the following subsections, we introduce the end-to-end workflow of analog RF computing-based edge AI. We first present the weight- and input-encoded RF waveforms that are inputs to the LO- and RF-ports of the passive mixer, respectively. We then discuss how their corresponding baseband waveforms should be generated for carrying out analog MVMs. Next, we describe the hardware characteristics of the passive mixer. We discuss its two operating regions and explain how they impact the computing accuracy. Finally, we derive a baseband equivalent model for analog RF computing, which forms the basis of the physical layer design problem in Section~\ref{sec:problem_formulation}. 

\subsection{Construction of RF Waveforms}\label{sec:mumimo_ports}
We consider the downlink of a MU-MIMO system where the BS employs $N_\text{t}$ transmit antennas and serves $K$ clients\footnote{The clients are chosen to match the low-complexity edge devices targeted by this work, such as surveillance cameras, wearables, and industrial sensors. These are representative use cases of the fifth-generation (5G) reduced-capability (RedCap) user equipment \cite{3gpp_ts38306}, which, in Frequency Range 1 (FR1), typically supports only one downlink MIMO layer with one receive branch.} simultaneously. The downlink channel between the BS and client $k\in\mathcal{K}$ is denoted by $\mathbf{h}_k\in\mathbb{C}^{N_\text{t}}$. As illustrated in Fig.~\ref{fig:system_schematic}, for NN layer $\ell\in\mathcal{L}$, the BS generates a complex-valued weight-encoded baseband waveform $s_w^{(\ell)}(t)$ with continuous time $t$ after baseband processing, upconverts it to carrier frequency $f_w$ through I/Q modulation, and applies a broadcast beamforming vector $\mathbf{f}^{(\ell)}\in\mathbb{C}^{N_\text{t}}$ to send it to $K$ clients simultaneously. The beamforming gain at client $k$ for layer $\ell\in \mathcal{L}$ is denoted by $g_k^{(\ell)} \triangleq \mathbf{h}_k^{\mathsf H}\mathbf{f}^{(\ell)} \in \mathbb{C}$. 
Hence, after downlink propagation, the real-valued weight-encoded RF waveform received by client $k$ is given by
\begin{equation}
 v_{\mathrm{LO},k}^{(\ell)}(t) = \Re\big\{\sqrt{P_{w,0}}\,g_k^{(\ell)}\,s_w^{(\ell)}(t)\,e^{j2\pi f_w t}\big\}, 
\label{eq:lo_passband}
\end{equation}
which drives the LO port of the passive mixer, with $P_{w,0}$ being a fixed reference power level for the waveform.

In parallel, as shown in Fig.~\ref{fig:system_schematic}, client $k$ generates a complex-valued input-encoded waveform $s_{x,k}^{(\ell)}(t)$ and applies a scaling coefficient $\beta_k^{(\ell)}\in\mathbb{C}$ at the baseband before upconverting it to carrier frequency $f_x$ via I/Q modulation. Hence, the real-valued input-encoded waveform applied to the RF port of the passive mixer is given by
\begin{equation}
  v_{\mathrm{RF},k}^{(\ell)}(t) = \Re\big\{\sqrt{P_{x,0}}\,\beta_k^{(\ell)}\,s_{x,k}^{(\ell)}(t)\,e^{j2\pi f_x t}\big\},
\label{eq:rf_passband}
\end{equation}
where $P_{x,0}$ is a fixed reference power level for the waveform. 

Under the ideal case, the output at the IF port of the passive mixer is proportional to the product of the LO-port weight-encoded waveform and the RF-port input-encoded waveform. Although the passive mixer is driven by the two real-valued RF waveforms in~\eqref{eq:lo_passband} and \eqref{eq:rf_passband}, the analog RF computing process is most conveniently described using the corresponding complex-valued baseband waveforms, i.e., $s_w^{(\ell)}(t)$ and $s_{x,k}^{(\ell)}(t)$. We next discuss how to construct baseband waveforms to realize the analog RF computing-based MVM. We limit our introduction to the key steps essential for the physical layer design problem considered in this paper, and refer interested readers to \cite{wise} for detailed baseband implementations.

\subsection{Construction of Baseband Waveforms}\label{sec:tone_assignment}

The NN weights $\mathbf{W}^{(\ell)}$ may have a large dimension $M^{(\ell)}$, which may not be loaded at once onto a single orthogonal frequency division multiplexing (OFDM) symbol. Each layer $\ell$ is hence processed in smaller row blocks of size $M^{\prime(\ell)}$, as shown in Figs.~\ref{fig:waveform_gen}\subref{fig:waveform_gen_a} and \subref{fig:waveform_gen_b}. Each block has input length $N^{(\ell)}$ and output length $M^{\prime(\ell)}$. Let $Q^{(\ell)} \triangleq \left\lceil \frac{M^{(\ell)}}{M^{\prime(\ell)}} \right\rceil$ denote the number of row blocks for layer $\ell$. We represent the $b$-th row block of $\mathbf{W}^{(\ell)}$ as $\mathbf{W}_{[b]}^{(\ell)}\in\mathbb{C}^{M^{\prime(\ell)}\times N^{(\ell)}}$, where $b\in\{1,\dots,Q^{(\ell)}\}$. The corresponding output block at client $k$ is $\mathbf{y}_{k,[b]}^{(\ell)} = \mathbf{W}_{[b]}^{(\ell)}\mathbf{x}_k^{(\ell)}$. For clarity, the following waveform construction focuses on one representative block and omits the indices $k$, $\ell$, and $b$, so that the computing target becomes $\mathbf{y}=\mathbf{W}\mathbf{x}$ with $\mathbf{W}\in\mathbb{C}^{M'\times N}$ and $\mathbf{y}\in\mathbb{C}^{M'}$. To accommodate anti-alias (AA) filtering, we encode a slightly wider band of $\widetilde{M}$ entries, defined as $\widetilde{M}\triangleq (1+\vartheta)M^{\prime}$, where $\vartheta$ is a non-negative guard factor, as illustrated in Fig.~\ref{fig:waveform_gen}\subref{fig:waveform_gen_c}. Only $M^{\prime}$ of these entries carry useful MVM outputs.

\begin{figure}[t]
  \centering
  \subfloat[]{\includegraphics[width=0.14\textwidth]{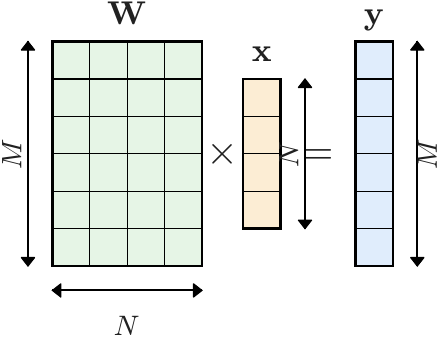}\label{fig:waveform_gen_a}} \,\,\,
  \subfloat[]{\includegraphics[width=0.14\textwidth]{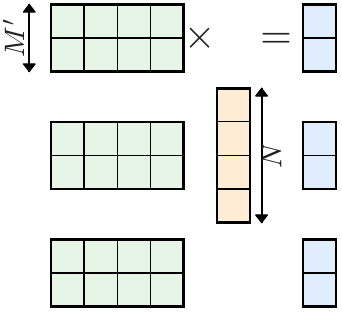}\label{fig:waveform_gen_b}} \,\,\,
  \subfloat[]{\includegraphics[width=0.14\textwidth]{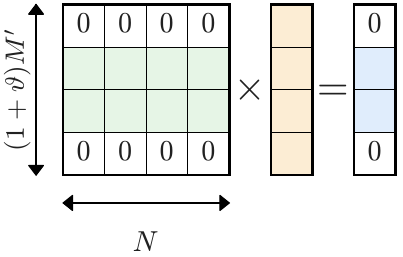}\label{fig:waveform_gen_c}}\\
  \subfloat[]{\includegraphics[width=0.22\textwidth]{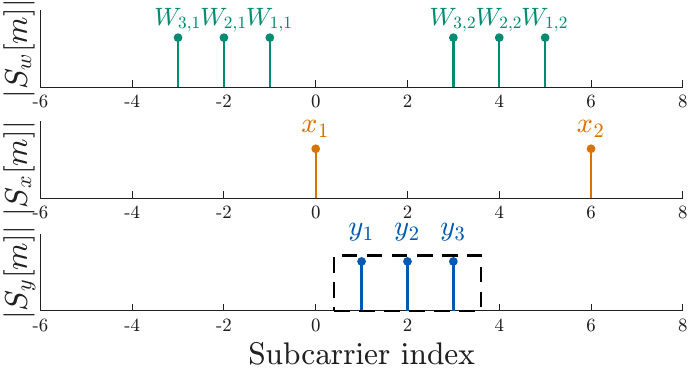}\label{fig:waveform_gen_d}} \,\,\,
  \subfloat[]{\includegraphics[width=0.22\textwidth]{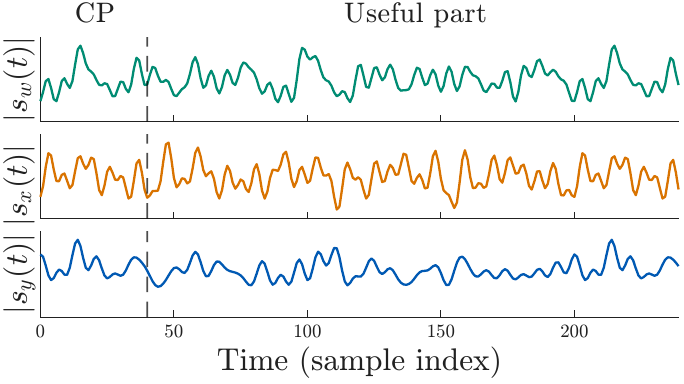}\label{fig:waveform_gen_e}}
  \caption{An illustration of the baseband waveform construction and subcarrier mapping for passive mixer-based analog MVM. }
  \label{fig:waveform_gen}
\end{figure}

We construct the useful baseband input-encoded and weight-encoded waveforms as
\begin{equation}
  s_x(t) = \sum_{n\in\mathcal{N}} x_n \, e^{j2\pi \nu_x(n)\Delta f t}, \,\, 0\le t<T_{\mathrm{u}}, \,\, \text{and}
  \label{eq:sx_waveform}
\end{equation}
\begin{equation}
  s_w(t) = \sum_{m\in\mathcal{M}^\prime}\sum_{n\in\mathcal{N}} W_{m,n} \, e^{j2\pi \nu_w(m,n)\Delta f t}, \,\, 0\le t<T_{\mathrm{u}},
  \label{eq:sw_waveform}
\end{equation}
respectively, in which $\mathcal{N}\triangleq\{1,\ldots,N\}$, $\mathcal{M}^\prime\triangleq\{1,\ldots,M^{\prime}\}$, $x_n$ is the $n$-th entry of the input vector $\mathbf{x}$, $W_{m,n}$ denotes the $(m,n)$-th entry of the weight matrix $\mathbf{W}$, $\nu_x(n)$ and $\nu_w(m,n)$ denote the integer subcarrier indices used for $x_n$ and $W_{m,n}$, respectively, $\Delta f$ is the subcarrier spacing, and $T_{\mathrm{u}}\triangleq 1/\Delta f$ denotes the useful OFDM symbol duration. According to \eqref{eq:sx_waveform} and \eqref{eq:sw_waveform}, $x_n$ is loaded onto the $\nu_x(n)$-th subcarrier of the input-encoded waveform, while $W_{m,n}$ is loaded onto the $\nu_w(m,n)$-th subcarrier of the weight-encoded waveform. To match the notation in Fig.~\ref{fig:waveform_gen}, we denote the fast Fourier transform (FFT) coefficients of $s_x(t)$ and $s_w(t)$ by $S_x[\cdot]$ and $S_w[\cdot]$, respectively, so that $S_x[\nu_x(n)] = x_n$ and $S_w[\nu_w(m,n)] = W_{m,n}$ on the loaded subcarriers. To realize analog MVM, $\nu_x(n)$ is set as
\begin{equation}
  \nu_x(n)\triangleq \nu_0+(n-1)\widetilde{M},\qquad n\in\mathcal{N},
  \label{eq:nux_rule}
\end{equation}
in which $\nu_0$ is the reference subcarrier location. The weight $W_{m,n}$ is assigned to the subcarrier with index $\nu_x(n)-m$, i.e., 
\begin{equation}
  \nu_w(m,n)\triangleq \nu_x(n)-m,\qquad m\in\mathcal{M}^\prime,~ n\in\mathcal{N}.
  \label{eq:nuw_rule}
\end{equation}
The expressions in \eqref{eq:sx_waveform} and \eqref{eq:sw_waveform} specify the loaded subcarrier content that realizes the analog MVM. In implementation, the corresponding OFDM symbols are mapped to the actual mixer-input waveforms under a fixed waveform-normalization convention before upconversion. This only introduces a known common gain and does not change the subcarrier relations in \eqref{eq:nux_rule} and \eqref{eq:nuw_rule}. For notational simplicity, we keep using $s_x(t)$ and $s_w(t)$ for the resulting normalized mixer-input waveforms.
Fig.~\ref{fig:waveform_gen}\subref{fig:waveform_gen_d} illustrates an example of the subcarrier loading step, after which a cyclic prefix (CP) of duration $T_{\mathrm{cp}}\triangleq \varpi T_{\mathrm{u}}$ is appended, where $\varpi $ denotes the CP overhead. This yields a total symbol duration of $(1+\varpi)T_{\mathrm{u}}$, as shown in Fig.~\ref{fig:waveform_gen}\subref{fig:waveform_gen_e}. 

The subcarrier loading rules in \eqref{eq:nux_rule} and \eqref{eq:nuw_rule} are motivated by the fact that the passive mixer performs pointwise multiplication of the LO-port and RF-port waveforms and outputs the result at the IF port. After CP removal, pointwise multiplication of the waveforms over the useful interval becomes circular convolution of their loaded subcarriers. After selecting the difference band around $f_y\triangleq f_x-f_w$, we denote the resulting IF-port baseband waveform by $s_y(t)$. The product of the pair $(x_{n},W_{m,n})$, for all $n\in\mathcal{N}$, always lands on the $m$-th subcarrier of $s_y(t)$ because $\nu_x(n)-\nu_w(m,n)=m$. 

Accordingly, the $m$-th FFT coefficient of $s_y(t)$, denoted by $S_y[m]$, satisfies
\begin{equation}
  S_y[m]\propto \sum_{n\in\mathcal{N}}W_{m,n}x_n,\qquad m\in\mathcal{M^\prime},
  \label{eq:sy_block}
\end{equation}
where $\propto$ denotes proportionality. Thus, $S_y[m]$ gives the desired analog MVM result for the current block up to a constant factor determined by both the mixer and the physical layer design. For client $k$ in layer $\ell$, repeating this procedure over all $Q^{(\ell)}$ row blocks yields the full layer output $\mathbf{y}_k^{(\ell)}$.

\subsection{Hardware Characteristics of Passive Mixers}\label{sec:device_regimes}
The baseband waveform construction above yields the desired MVM only when the passive mixer behaves approximately as a multiplier of two RF waveforms. However, this is not always true and depends on the operating region of the mixer. In the following, we discuss the two operating regions of passive mixers and then present a baseband equivalent model for analog RF computing, which is key to the physical layer design to be discussed later.

\begin{figure}[t]
  \centering
  \includegraphics[width=0.776\columnwidth,trim=0bp 0bp 637bp 0bp,clip]{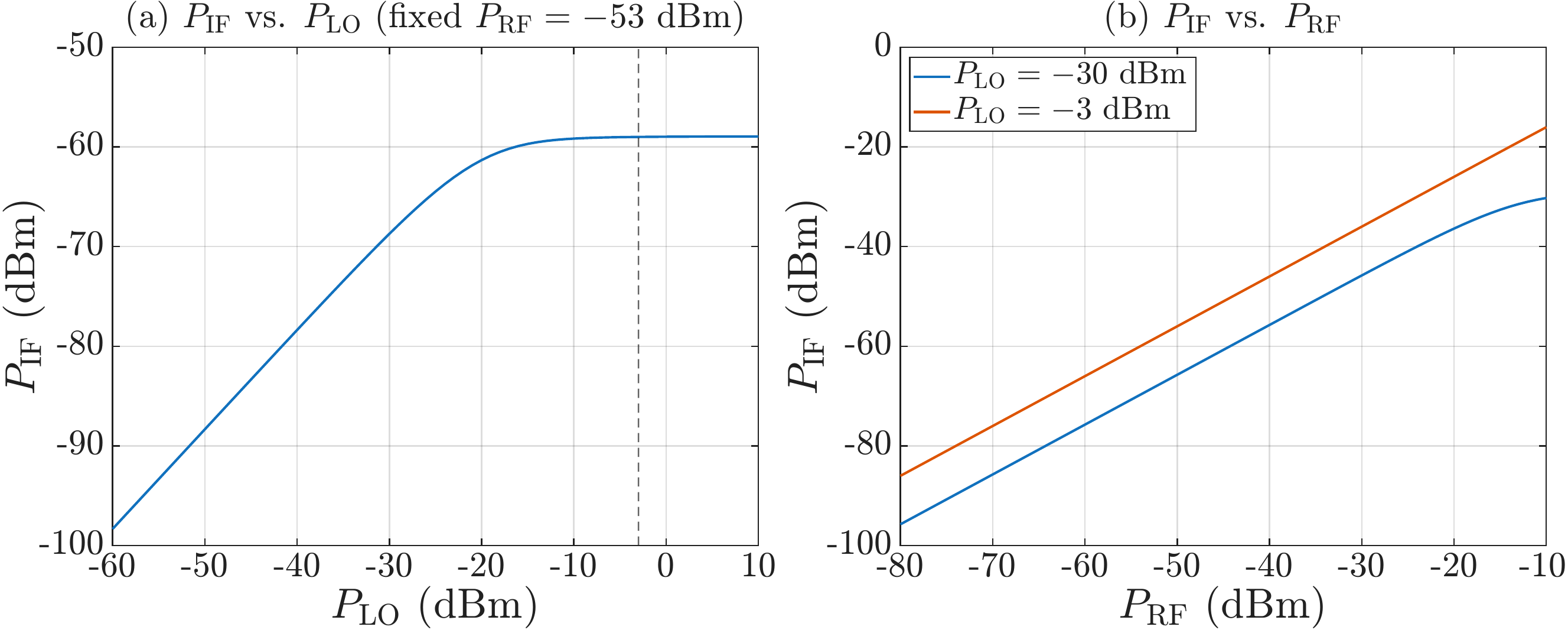}\par\smallskip
  \includegraphics[width=0.776\columnwidth,trim=637bp 0bp 0bp 0bp,clip]{./figures/fig3_mixer_regions.pdf}
  \caption{IF-port output power scaling versus LO- and RF-port input powers under a calibrated diode-ring mixer model \cite{zem4300_datasheet}. (a) shows the transition from the small-signal region to the LO-saturated region, and the vertical line indicates the transition point used to define $a_{\max}$. $P_\mathrm{IF}$ first grows proportionally with  $P_\mathrm{LO}$ and then saturates. (b) shows that $P_\mathrm{IF}$ grows proportionally with $P_\mathrm{RF}$. }
  \label{fig:mixer_regions}
\end{figure}

\subsubsection{Two Operating Regions of Passive Mixers}
The passive mixer considered in this paper changes its input-output behavior with the received LO drive \cite{razavi2011rf,zem4300_datasheet}. For a fixed RF-port power $P_{\mathrm{RF}}$, Fig.~\ref{fig:mixer_regions}(a) shows that the IF-port power $P_{\mathrm{IF}}$ first grows approximately linearly with the LO-port power $P_{\mathrm{LO}}$ and then flattens. We refer to the former as the \emph{small-signal} region and to the latter as the \emph{LO-saturated} region. Fig.~\ref{fig:mixer_regions}(b) shows that, for a fixed LO-port power, $P_{\mathrm{IF}}$ remains linear in $P_{\mathrm{RF}}$.

In the small-signal region, the output of the mixer at the IF port behaves approximately as a multiplier of LO-port and RF-port waveforms in~\eqref{eq:lo_passband} and \eqref{eq:rf_passband}, i.e., 
\begin{equation}
  v_{\mathrm{IF},k}^{(\ell)}(t)\approx \sqrt{\rho_{\mathrm{mixer}}}\,v_{\mathrm{LO},k}^{(\ell)}(t)\,v_{\mathrm{RF},k}^{(\ell)}(t),
  \label{eq:miwen_mixer_smallsig}
\end{equation}
where $\rho_{\mathrm{mixer}}$ denotes the mixer conversion coefficient. This yields $P_{\mathrm{IF}}\propto P_{\mathrm{LO}}P_{\mathrm{RF}}$.
In the LO-saturated region, the effective LO waveform becomes a two-level switching waveform with amplitude $V_{\mathrm{LO,sat}}>0$, i.e., $V_{\mathrm{LO,sat}}\,\mathrm{sgn}\!\big(v_{\mathrm{LO},k}^{(\ell)}(t)\big)$, and the output of the mixer at the IF port is
\begin{equation}
  v_{\mathrm{IF},k}^{(\ell)}(t)\approx \sqrt{\rho_{\mathrm{mixer}}}\,V_{\mathrm{LO,sat}}\,v_{\mathrm{RF},k}^{(\ell)}(t)\,\mathrm{sgn}\!\big(v_{\mathrm{LO},k}^{(\ell)}(t)\big), 
  \label{eq:miwen_mixer_switch}
\end{equation}
where $\mathrm{sgn}(\cdot)$ denotes the sign function. 

LO saturation can be undesirable for analog RF computing for two reasons. First, once the LO port enters saturation, further increasing the LO-port power yields little increase in the power of the recovered MVM result. Since analog RF computing is ultimately limited by noise, any additional LO power is then wasted. Second, LO saturation changes the LO waveform that participates in the analog multiplication. In the LO-saturated region, the mixer is effectively driven by a hard-limited waveform $\mathrm{sgn}(v_{\mathrm{LO},k}^{(\ell)}(t))$ instead of the original $v_{\mathrm{LO},k}^{(\ell)}(t)$. To illustrate the consequence, suppose the original LO waveform contains two subcarriers at indices $p$ and $q$. After hard limiting, the waveform is no longer composed of only these two tones, and additional tones appear at other indices such as $2p-q$ and $2q-p$. When this distorted LO waveform multiplies the RF-port waveform, these extra LO tones shift part of the RF spectrum to unintended IF subcarriers. Some of those components can land on the same IF subcarriers that carry the desired MVM outputs. They therefore appear as in-band distortion and cannot be removed by AA filtering. To keep the model simple while still reflecting this operating boundary, we replace the received coefficient $g_k^{(\ell)}$ by its clipped version
\begin{equation}
  a_k^{(\ell)} \triangleq \min\big\{|g_k^{(\ell)}|,a_{\max}\big\}\,e^{j\angle g_k^{(\ell)}}.
\label{eq:ak_clip}
\end{equation}
That is, $a_k^{(\ell)}$ preserves the phase of $g_k^{(\ell)}$ while limiting its magnitude to the saturation boundary $a_{\max}$. Here $a_{\max}>0$ is a calibrated saturation boundary. The received LO-port power is $P_{\mathrm{LO},k}^{(\ell)}=P_{w,0}|g_k^{(\ell)}|^2$. Thus, setting an LO-drive threshold $P_{\mathrm{LO,th}}$ corresponds to
\begin{equation}
  P_{w,0}a_{\max}^2 = P_{\mathrm{LO,th}}, \qquad a_{\max}=\sqrt{\tfrac{P_{\mathrm{LO,th}}}{P_{w,0}}},
  \label{eq:amax_relation}
\end{equation}
and enforcing the beamforming gain $|g_k^{(\ell)}|\le a_{\max}$ keeps the mixer operating in the small-signal region. 

\subsection{Baseband Equivalent Model for Analog RF Computing}\label{sec:compute_model}
We will focus on the small-signal region, because only in this region does the passive mixer behave approximately as a direct multiplier. Combining \eqref{eq:miwen_mixer_smallsig} with the baseband waveform construction in Section~\ref{sec:tone_assignment}, each decoded IF subcarrier is proportional to the corresponding elements of the MVM result. After stacking $Q^{(\ell)}$ decoded output blocks of layer $\ell$, we model the output of layer $\ell$ at client $k$ as
\begin{equation}
  \widetilde{\mathbf{y}}_k^{(\ell)} = a_k^{(\ell)}\,\mathbf{W}^{(\ell)}\beta_k^{(\ell)}\mathbf{x}_k^{(\ell)} + \mathbf{n}_k^{(\ell)},
  \label{eq:ytilde}
\end{equation}
where $\mathbf{n}_k^{(\ell)}\sim \mathcal{CN}(\mathbf{0},\sigma_{0}^{2,(\ell)}\mathbf{I}_{M^{(\ell)}})$ is an effective additive term capturing the thermal noise introduced in the process of analog RF computing, which is detailed in Section~\ref{sec:Computing_Accuracy_Metric}. The error of analog computing is mainly determined by the ratio between the signal power of analog MVM results and the noise power. Equation~\eqref{eq:ytilde} is the baseband equivalent model for analog RF computing in the small-signal region. 

From~\eqref{eq:ytilde}, we can see that the controllable part of the analog computing error is mainly the end-to-end gain $a_k^{(\ell)}\beta_k^{(\ell)}$ since the noise is only determined by the hardware and the environment. This observation is central for the physical layer design. The BS-side beamformer determines $a_k^{(\ell)}$ and the received amplitude of the LO-port waveform, while the client-side scaling determines the RF-port amplitude through $\beta_k^{(\ell)}$. Because their effects are multiplicative, the same computing accuracy can be achieved by different tradeoffs of the energy consumption at the BS and the client sides. To achieve the same computing accuracy, increasing the BS beamforming gain can reduce the required client scaling, and vice versa. Depending on different BS-side and client-side energy budgets, different points on this tradeoff become desirable. This tradeoff is constrained, since both the BS-side and the client side transmit powers are limited, and the mixer small-signal region is determined by the specific hardware. Therefore, the physical layer design must study the joint design of BS-side beamformer and client-side scaling subject to computing accuracy, transmit power, and hardware constraints. In the next section, we derive explicit expressions for the computing accuracy and inference energy consumption as functions of the beamforming vector $\mathbf{f}^{(\ell)}$ and scaling coefficient $\beta_k^{(\ell)}$, and then formulate the corresponding optimization problem.

\textit{Remark:} We mainly focus on resource-constrained edge devices requiring intelligent services with limited power supply and computing capability. The intended edge devices are stationary or slowly moving and perform a predefined set of tasks, such as image classification. Examples include low-end surveillance cameras and forest fire sensors. A large number of devices serving the same task can share a common NN model, allowing the BS to broadcast its weights over MU-MIMO systems while each client processes its own local input.

\section{Performance Metrics and Problem Formulation}\label{sec:problem_formulation}

In this section, we derive the performance metrics of analog RF computing as functions of the BS-side beamformer and client-side scaling coefficients. We first quantify analog computing errors by using normalized mean square error (NMSE) and then introduce an inference energy consumption model that captures BS-side beamforming, client-side scaling, and readout energy consumption. We then formulate a joint BS beamforming and client-scaling problem for the physical layer design of analog RF computing in MU-MIMO wireless systems, which can satisfy per-client computing accuracy constraints with the minimum inference energy consumption. 

\subsection{Computing Accuracy Metric}\label{sec:Computing_Accuracy_Metric}

Following from \eqref{eq:ytilde}, client $k$ estimates the end-to-end gain $a_k^{(\ell)}\beta_k^{(\ell)}$ based on pilots and applies scalar equalization. This gives $\widehat{\mathbf{y}}_k^{(\ell)} \triangleq \widetilde{\mathbf{y}}_k^{(\ell)}/(a_k^{(\ell)}\beta_k^{(\ell)})=\mathbf{y}_k^{(\ell)}+\widetilde{\mathbf{n}}_k^{(\ell)}$, where $\widetilde{\mathbf{n}}_k^{(\ell)}\sim\mathcal{CN}(\mathbf{0},\sigma_{k,\mathrm{eq}}^{2,(\ell)}\mathbf{I}_{M^{(\ell)}})$ and $\sigma_{k,\mathrm{eq}}^{2,(\ell)}\triangleq\sigma_{0}^{2,(\ell)}/(|a_k^{(\ell)}|^2|\beta_k^{(\ell)}|^2)$. The post-equalization NMSE of the analog MVM is
\begin{equation}
  \mathrm{NMSE}_{k}^{(\ell)}
  \triangleq
  \frac{\mathbb{E}[\|\widehat{\mathbf{y}}_k^{(\ell)}-\mathbf{y}_k^{(\ell)}\|_2^2]}{\mathbb{E}[\|\mathbf{y}_k^{(\ell)}\|_2^2]}.
  \label{eq:nmse_def}
\end{equation}
Denote $\bar S_k^{(\ell)}\triangleq \mathbb{E}[\|\mathbf{W}^{(\ell)}\mathbf{x}_k^{(\ell)}\|_2^2]/M^{(\ell)}$ as the average actual output power per decoded entry, where the expectation is taken over the inputs of client $k$ in layer $\ell$ and the channel. Then we have $\mathrm{NMSE}_{k}^{(\ell)}=\sigma_{k,\mathrm{eq}}^{2,(\ell)}/\bar S_k^{(\ell)}$, i.e., the analog computing error is determined by the noise power relative to the actual output power. Notice that the NMSE metric depends on $\bar S_k^{(\ell)}$, which changes with the input to the current layer at each client, but the BS does not know the client-side input during transmission. Instead, online optimization can only use input statistics from an offline calibration dataset. We therefore replace $\bar S_k^{(\ell)}$ with a reference output power $\bar S_{\mathrm{ref}}^{(\ell)}\triangleq \mathbb{E}_{\mathrm{cal}}[\|\mathbf{W}^{(\ell)}\mathbf{x}^{(\ell)}\|_2^2]/M^{(\ell)}$, where $\mathbb{E}_{\mathrm{cal}}$ denotes the expectation over the offline calibration dataset. This layer-specific reference output power is common to all clients. The reference NMSE is given by $\mathrm{NMSE}_{k,\mathrm{ref}}^{(\ell)}\triangleq\sigma_{0}^{2,(\ell)}/(|a_k^{(\ell)}|^2|\beta_k^{(\ell)}|^2\bar S_{\mathrm{ref}}^{(\ell)})$.

Under the mixer model in Section~\ref{sec:device_regimes}, the LO-port and RF-port powers at client $k$ for layer $\ell$ are $P_{\mathrm{LO},k}^{(\ell)}=P_{w,0}|g_k^{(\ell)}|^2$ and $P_{\mathrm{RF},k}^{(\ell)}=P_{x,0}|\beta_k^{(\ell)}|^2$, respectively. The small-signal multiplier model in \eqref{eq:miwen_mixer_smallsig} implies that the desired output power scales with their product, while the noise power over bandwidth $B$ scales with $k_B T_0 B/\rho_{\mathrm{nf}}$. Here, $k_B$ is the Boltzmann's constant, $T_0$ is the environmental temperature, $B$ is the occupied bandwidth, and $\rho_{\mathrm{nf}}$ is the effective receiver noise coefficient. This leads to the reference SNR constant $\kappa\triangleq \rho_{\mathrm{mixer}}\rho_{\mathrm{nf}}P_{w,0}P_{x,0}/(k_B T_0 B)$. We calibrate the noise power in the baseband equivalent model by setting $\sigma_{0}^{2,(\ell)}=\bar S_{\mathrm{ref}}^{(\ell)}/\kappa$. Substitution gives the online reference metric of NMSE, i.e.,
\begin{equation}
  \mathrm{NMSE}_{k,\mathrm{ref}}^{(\ell)} = \frac{1}{\kappa\,|a_k^{(\ell)}|^2\,|\beta_k^{(\ell)}|^2},
  \label{eq:nmse_abeta}
\end{equation}
which satisfies that $\mathrm{NMSE}_{k}^{(\ell)}=(\bar S_{\mathrm{ref}}^{(\ell)}/\bar S_k^{(\ell)})\mathrm{NMSE}_{k,\mathrm{ref}}^{(\ell)}$. When solving the physical layer design problem online, the accuracy constraint is imposed on \eqref{eq:nmse_abeta}, while $\mathrm{NMSE}_{k}^{(\ell)}$ is used only to evaluate the performance in simulations.

\subsection{Inference Energy Consumption Metric}\label{sec:energy_model}
We next derive the inference energy consumption metric of analog RF computing-based edge AI, which mainly includes BS-side transmission and client-side waveform generation and readout. Since the diode-ring mixer is a passive device, it does not contribute directly to the energy consumption. As analog RF computing is realized by OFDM-style blocks with FFT-based readout, the client-side energy per multiply-accumulate (MAC),\footnote{\minorchange{Throughout this paper, MAC by default refers to a \textit{real-valued} MAC. For example, an $M^{(\ell)}\times N^{(\ell)}$ complex-valued MVM requires $4M^{(\ell)}N^{(\ell)}$ real-valued MACs.}} in the unit of J/MAC, can be decomposed into three parts, i.e., $e_{\mathrm{client},k}^{(\ell)}=e_{1,k}^{(\ell)}+e_2^{(\ell)}+e_3^{(\ell)}$, where $e_{1,k}^{(\ell)}$ is the waveform generation energy per MAC, $e_2^{(\ell)}$ is the ADC sampling energy per MAC, and $e_3^{(\ell)}$ is the digital decoding energy per MAC. The first term $e_{1,k}^{(\ell)}$ changes with the scaling coefficient $\beta_k^{(\ell)}$, while the last two terms are constant once the receiver and waveform design are fixed.

For layer $\ell$, each inference involves $P^{(\ell)}$ input vectors.\footnote{For a fully connected layer, the input is one vector and $P^{(\ell)}=1$. For a convolutional layer, the same convolution kernels are applied at many spatial windows of the input feature map. Image to column operation unfolds each window into one input vector and multiplies the same weight matrix by all unfolded vectors. $P^{(\ell)}$ is equal to the number of convolution windows \cite{sze2017efficient}.} \minorchange{Let $T_w^{(\ell)}$ denote the duration of the weight waveform of layer $\ell$ used to calculate the BS transmission energy and the client-side waveform-generation energy. For FFT length $L_{\mathrm{fft}}$ and bandwidth $B$, the duration of each symbol including its CP is $(1+\varpi^{(\ell)})L_{\mathrm{fft}}/B$. Thus,}
\begin{equation*}
\minorchange{T_w^{(\ell)}=\frac{L_{\mathrm{fft}}}{B}P^{(\ell)}Q^{(\ell)}\left\lceil\frac{N^{(\ell)}\widetilde{M}^{(\ell)}}{L_{\mathrm{fft}}}\right\rceil(1+\varpi^{(\ell)}).}
\end{equation*}

\minorchange{We next derive the waveform generation, ADC sampling, and digital decoding energies per MAC, denoted by $e_{1,k}^{(\ell)}$, $e_2^{(\ell)}$, and $e_3^{(\ell)}$, respectively. Each client generates its input-encoded waveform. The client-side waveform generation energy per MAC} $e_{1,k}^{(\ell)}$ scales linearly with the RF-port input power $P_{\mathrm{RF},k}^{(\ell)}=P_{x,0}|\beta_k^{(\ell)}|^2$ and inversely with the RF chain efficiency $\rho_{\mathrm{radio}}$, as given by
\begin{equation}
  e_{1,k}^{(\ell)} = \minorchange{T_w^{(\ell)}\frac{P_{x,0}|\beta_k^{(\ell)}|^2}{\rho_{\mathrm{radio}}}\frac{1}{4M^{(\ell)}N^{(\ell)}P^{(\ell)}}.}
  \label{eq:e1_beta}
\end{equation}
For the receiver, let $e_{\mathrm{adc}}$ denote the energy per ADC sample and $e_{\mathrm{dig}}$ denote the digital decoding energy per MAC. Because layer $\ell$ is read out over $Q^{(\ell)}$ blocks, each with $\widetilde{M}^{(\ell)}=(1+\vartheta^{(\ell)})M^{\prime(\ell)}$ IF-port subcarriers, the energy terms per MAC are $e_2^{(\ell)}=\frac{1+\vartheta^{(\ell)}}{2N^{(\ell)}}e_{\mathrm{adc}}$ and $e_3^{(\ell)}=\frac{1+\vartheta^{(\ell)}}{2N^{(\ell)}}\log_2((1+\vartheta^{(\ell)})M^{\prime(\ell)})e_{\mathrm{dig}}$, respectively~\cite{wise}.

The client energy and BS energy, in the unit of J, in layer $\ell$ are given by
\begin{equation}
\begin{aligned}
  & E_{\mathrm{client},k}^{(\ell)} = \minorchange{4}M^{(\ell)} N^{(\ell)} P^{(\ell)} e_{\mathrm{client},k}^{(\ell)}, \\
  & E_{\mathrm{BS}}^{(\ell)} = \frac{T_w^{(\ell)}}{\eta_{\mathrm{BS}}}\,P_{w,0}\|\mathbf{f}^{(\ell)}\|_2^2,
\end{aligned}
  \label{eq:layer_energies}
\end{equation}
where $\eta_{\mathrm{BS}}\in(0,1]$ is the BS efficiency.

From \eqref{eq:layer_energies}, the average BS-side energy per MAC $\bar e_{\mathrm{BS}}$ and client-side energy per MAC $\bar e_{\mathrm{client}}$ are respectively defined as
\begin{equation}
\begin{aligned}
  & \bar e_{\mathrm{BS}} \triangleq
  \frac{\sum_{\ell=1}^{L} E_{\mathrm{BS}}^{(\ell)}}
  {\minorchange{4}K\sum_{\ell=1}^{L} M^{(\ell)}N^{(\ell)}P^{(\ell)}}, \\
  & \bar e_{\mathrm{client}} \triangleq
  \frac{\sum_{k=1}^{K}\sum_{\ell=1}^{L}E_{\mathrm{client},k}^{(\ell)}}
  {\minorchange{4}K\sum_{\ell=1}^{L} M^{(\ell)}N^{(\ell)}P^{(\ell)}},
\end{aligned}
\end{equation}
where ${\minorchange{4}K\sum_{\ell=1}^{L} M^{(\ell)}N^{(\ell)}P^{(\ell)}}$ is the total number of MACs at all clients.
The client-side energy per MAC can be decomposed as $\bar e_{\mathrm{client}}=\bar e_1+\bar e_2+\bar e_3$. The average energies per MAC for waveform generation $\bar e_1$, ADC sampling $\bar e_2$, and digital decoding $\bar e_3$ are respectively given by
\begin{equation}
\begin{aligned}
  & \bar e_1 \triangleq
  \frac{\sum_{k=1}^{K}\sum_{\ell=1}^{L} M^{(\ell)}N^{(\ell)}P^{(\ell)}e_{1,k}^{(\ell)}}
  {K\sum_{\ell=1}^{L} M^{(\ell)}N^{(\ell)}P^{(\ell)}}, \\
  & \bar e_j \triangleq
  \frac{\sum_{\ell=1}^{L} M^{(\ell)}N^{(\ell)}P^{(\ell)}e_j^{(\ell)}}
  {\sum_{\ell=1}^{L} M^{(\ell)}N^{(\ell)}P^{(\ell)}},\quad j\in\{2,3\}.
\end{aligned}
\label{eq:average_energy_components}
\end{equation}

\subsection{Accuracy-Constrained Energy Minimization Problem}\label{sec:accuracy_design}
We formulate the physical layer design problem to minimize the inference energy consumption of analog RF computing-based edge inference while satisfying the computing accuracy constraints, the transmit power constraints, and the hardware constraints. We parameterize the computing accuracy targets by a root NMSE target $\epsilon_k^{(\ell)}$ and impose $\mathrm{NMSE}_{k,\mathrm{ref}}^{(\ell)} \le (\epsilon_k^{(\ell)})^2$, which is equivalent to a lower bound on the end-to-end analog multiplication gain, i.e., 
\begin{equation}
  |a_k^{(\ell)}\beta_k^{(\ell)}| \ge u_k^{(\ell)},\quad
  u_k^{(\ell)} \triangleq \sqrt{\frac{1}{{\kappa\cdot}(\epsilon_k^{(\ell)})^2}}, \quad \forall k \in \mathcal{K}\;\forall \ell \in \mathcal{L}.
  \label{eq:u_req_def}
\end{equation}
The designed algorithm should support different values of $\epsilon_k^{(\ell)}$ for different layers $\ell$ and clients $k$. The beamforming coefficient is $a_k^{(\ell)}=g_k^{(\ell)}=\mathbf{h}_k^{\mathsf{H}}\mathbf{f}^{(\ell)}$ when the small-signal operating region constraint of the mixer is enforced, i.e., $|g_k^{(\ell)}|\le a_{\max}$, which depends on the hardware characteristics of the passive mixer. In addition, the transmit power should not exceed the limit at both the BS and the client sides, leading to constraints $P_{w,0}\|\mathbf{f}^{(\ell)}\|_2^2 \le P_{w,\max}$ and $|\beta_k^{(\ell)}|\le\beta_{\max}\triangleq\sqrt{P_{x,\max}/P_{x,0}}$.

For layer $\ell$, the physical layer variables are the BS beamformer $\mathbf{f}^{(\ell)}$ and the client scaling coefficients $\{\beta_k^{(\ell)}\}_{k=1}^K$. Among the client-side energy terms, only the waveform-generation term $e_{1,k}^{(\ell)}$ depends on these variables. The ADC sampling and digital decoding terms $e_2^{(\ell)}$ and $e_3^{(\ell)}$ are added back when reporting the energy metrics above, but they do not change the optimizer. We define the client-side waveform generation coefficient as
\begin{equation}
  {c}^{(\ell)} \triangleq
  \minorchange{\frac{T_w^{(\ell)}P_{x,0}}{\rho_{\mathrm{radio}}}.}
  \label{eq:ckl_def}
\end{equation}

Each layer of the NN weights is transmitted by using its own physical layer parameters $\mathbf{f}^{(\ell)}$ and $\{\beta_k^{(\ell)}\}_{k=1}^K$. No objective value or constraint couples two different layers, and the multi-layer design is decomposed into $L$ independent subproblems. To study the tradeoff between BS-side and client-side energy, we introduce a weighting parameter $\lambda\in[0,1]$. The objective function is a weighted sum of the BS-side and client-side energy consumption, i.e., $\lambda E_{\mathrm{BS}}^{(\ell)} + (1-\lambda) \sum_{k=1}^K E_{\mathrm{client},k}^{(\ell)}$. For problem P1 and throughout Section~\ref{sec:solver}, we consider an arbitrary layer and suppress its superscript $(\ell)$. After expanding the objective function, the per-layer physical layer design problem is as follows:
\begin{align}
  \text{P1}~~
  \underset{\mathbf{f},\{\beta_k\}_{k=1}^K}{\text{minimize}}~~ &
  \lambda\,\frac{T_w}{\eta_{\mathrm{BS}}}\,P_{w,0}\|\mathbf{f}\|_2^2
  + (1-\lambda)\sum_{k=1}^{K} c\,|\beta_k|^2 \notag\\
  \text{subject to}~~
  & \text{C1: } |a_k\beta_k| \ge u_k, ~~\forall k \in \mathcal{K}, \notag\\
  & \text{C2: } P_{w,0}\|\mathbf{f}\|_2^2 \le P_{w,\max}, \label{eq:Pl_layer}\\
  & \text{C3: } |\mathbf{h}_k^{\mathsf{H}}\mathbf{f}| \le a_{\max}, ~~\forall k \in \mathcal{K}, \notag\\
  & \text{C4: } |\beta_k| \le \beta_{\max}, ~~\forall k \in \mathcal{K} . \notag
\end{align}
The problem is nonconvex since constraint C1 contains the product of the BS beamforming gain and the client scaling. In the following section, we develop an algorithm for this generic per-layer problem and apply the same procedure to each layer.

\section{Algorithm Design for Analog RF Computing}\label{sec:solver}
In this section, we develop an efficient algorithm for solving the per-layer problem P1 in \eqref{eq:Pl_layer}. The main difficulty lies in the nonconvex coupling between the BS beamformer and the client scaling coefficients in constraint C1. To address this, we first derive the optimal client scaling in a closed form, then transform the problem to beamforming-only form, and finally solve the resulting problem by SCA techniques. By further utilizing the channel-subspace structure, the proposed algorithm operates in a low-dimensional subspace and remains scalable even when the BS has a massive number of antennas. 

\subsection{Designing Client-Side Scaling Coefficients}\label{sec:optimal_scaling}
For a fixed BS beamformer $\mathbf{f}$, the client-side scaling coefficients $\{\beta_k\}_{k=1}^K$ decouple across clients in both objective and constraints. We optimize each $\beta_k$ separately in a closed form.

\begin{lemma}[Energy-minimizing $\beta_k$ under C1--C4]\label{lem:beta_structure}
For a fixed beamformer $\mathbf{f}$, problem P1 at client $k$ can be reformulated as
\begin{equation*}
  \underset{\beta_k}{\text{minimize}}~|\beta_k|^2
  ~~\text{subject to}~~
  |a_k\beta_k| \ge u_k,\quad |\beta_k|\le\beta_{\max}. \label{eq:beta_subprob}
\end{equation*}
This problem is feasible if and only if $|a_k|\ge u_k/\beta_{\max}$.
When feasible, an energy-minimizing solution is
\begin{equation}
  \beta_k^\star = \frac{u_k}{a_k}. \label{eq:beta_opt_form}
\end{equation}
\end{lemma}

\begin{proof}
We first express $a_k=|a_k|e^{j\phi}$. The constraint depends only on $|\beta_k|$ through $|a_k||\beta_k|$.
Hence, an energy-minimizing solution aligns $\beta_k$ to cancel the phase of $a_k$ and attains constraint C1 with equality
$\beta_k=\alpha e^{-j\phi}$, where $\alpha\triangleq u_k/|a_k|$. This choice is feasible under constraint C4 if and only if $\alpha\le\beta_{\max}$, i.e., $|a_k|\ge u_k/\beta_{\max}$.
\end{proof}

\subsection{Designing BS-Side Beamformer}\label{sec:beamforming_reform}
Substituting \eqref{eq:beta_opt_form} into the second term of the objective function yields
\begin{equation}
  \sum_{k=1}^{K} c|\beta_k^\star|^2
  = \sum_{k=1}^{K}\frac{cu_k^2}{|a_k|^2}
  = \sum_{k=1}^{K}\frac{cu_k^2}{|\mathbf{h}_k^{\mathsf{H}}\mathbf{f}|^2}. \label{eq:obj_var_layer}
\end{equation}
Moreover, combining constraints C1 and C4 gives an equivalent feasibility constraint on the BS-side beamformer $\mathbf{f}$, i.e.,
\[
  \big|\mathbf{h}_k^{\mathsf{H}}\mathbf{f}\big| \ge u_k/\beta_{\max},\quad \forall k \in \mathcal{K},
\]
which we denote by constraint $\overline{\text{C1}}$. Therefore, problem P1 in \eqref{eq:Pl_layer} reduces to the following beamforming problem:
\begin{equation}\label{eq:P2}
\begin{aligned}
  \text{P2}~~
  \underset{\mathbf{f}}{\text{minimize}}~~ &
  \lambda\,\frac{T_w}{\eta_{\mathrm{BS}}}\,P_{w,0}\|\mathbf{f}\|_2^2
  + (1-\lambda)\sum_{k=1}^{K} \frac{cu_k^2}{|\mathbf{h}_k^{\mathsf{H}}\mathbf{f}|^2} \\
  \text{subject to}~~
  & \overline{\text{C1}}\mathpunct{:} |\mathbf{h}_k^{\mathsf{H}}\mathbf{f}|^2 \ge u_k^2/\beta_{\max}^2, ~~\forall k \in \mathcal{K}, \\
  & \text{C2: } P_{w,0}\|\mathbf{f}\|_2^2 \le P_{w,\max}, \\
  & \text{C3: } |\mathbf{h}_k^{\mathsf{H}}\mathbf{f}| \le a_{\max}, ~~\forall k \in \mathcal{K} .
\end{aligned}
\end{equation}
However, problem P2 in~\eqref{eq:P2} is still nonconvex because of the reciprocal terms in the objective. We propose to utilize SCA techniques to tackle P2 in the following. 

\subsubsection{SCA Beamforming Update}\label{sec:sca_update}
Recall that $a_k \triangleq \mathbf{h}_k^{\mathsf{H}}\mathbf{f}$ in the generic per-layer problem, to make the variable $\mathbf{f}$ explicit, we further define $a_k(\mathbf{f})\triangleq \mathbf{h}_k^{\mathsf{H}}\mathbf{f}$.
The mapping $\mathbf{f}\mapsto |a_k(\mathbf{f})|^2$ is convex in $\mathbf{f}$. Hence, for any current iterate $\mathbf{f}^{(i)}$, where $i$ denotes the SCA iteration index, the first-order Taylor expansion provides a global affine lower bound as
\begin{equation}\label{eq:uk_affine_lb}
\begin{aligned}
|a_k(\mathbf{f})|^2 \ge\;& |a_k^{(i)}|^2 \\
&+ 2\Re\!\left\{(a_k^{(i)})^\ast\big(a_k(\mathbf{f})-a_k^{(i)}\big)\right\} \\
\triangleq\;& \underline{u}_k^{(i)}(\mathbf{f}),
\end{aligned}
\end{equation}
where $a_k^{(i)}\triangleq a_k(\mathbf{f}^{(i)})$. Since $\underline{u}_k^{(i)}(\mathbf{f}) \le |a_k(\mathbf{f})|^2$ for all $\mathbf{f}$ and $\underline{u}_k^{(i)}(\mathbf{f}^{(i)})=|a_k^{(i)}|^2$,
and because $x\mapsto 1/x$ is convex and decreasing on $x>0$, we obtain the SCA upper bound as
\begin{equation}
  \frac{cu_k^2}{|a_k(\mathbf{f})|^2} \le \frac{cu_k^2}{\underline{u}_k^{(i)}(\mathbf{f})}, \,\, \forall \mathbf{f} \quad \text{subject to } \underline{u}_k^{(i)}(\mathbf{f})>0, \label{eq:recip_majorize}
\end{equation}
with equality at $\mathbf{f}=\mathbf{f}^{(i)}$.
At iteration $i$, we solve the following convex surrogate of problem P2 in \eqref{eq:P2}: 
\begin{align}
  \text{P2-SCA}~~
  \underset{\mathbf{f}}{\text{minimize}}~~ &
  \lambda\,\frac{T_w}{\eta_{\mathrm{BS}}}\,P_{w,0}\|\mathbf{f}\|_2^2
  + (1-\lambda)\sum_{k=1}^{K} \frac{cu_k^2}{\underline{u}_k^{(i)}(\mathbf{f})} \notag\\
  \text{subject to}~~
  & \widetilde{\overline{\text{C1}}}\mathpunct{:}~ \underline{u}_k^{(i)}(\mathbf{f}) \ge u_k^2/\beta_{\max}^2, ~~\forall k \in \mathcal{K}, \notag\\
  & \text{C2: } P_{w,0}\|\mathbf{f}\|_2^2 \le P_{w,\max}, \label{eq:P2_SCA}\\
  & \text{C3: } |\mathbf{h}_k^{\mathsf{H}}\mathbf{f}| \le a_{\max}, ~~\forall k \in \mathcal{K} . \notag
\end{align}
Constraint $\widetilde{\overline{\text{C1}}}$ has two roles. First, it makes each denominator $\underline{u}_k^{(i)}(\mathbf{f})$ in the surrogate objective positive. Second, since $\underline{u}_k^{(i)}(\mathbf{f})\le |a_k(\mathbf{f})|^2$, any point satisfying $\underline{u}_k^{(i)}(\mathbf{f})\ge u_k^2/\beta_{\max}^2$ also satisfies the original $\overline{\text{C1}}$ constraint $|a_k(\mathbf{f})|^2\ge u_k^2/\beta_{\max}^2$. 
Let $F(\mathbf{f})$ denote the objective of \eqref{eq:P2} and let $Q^{(i)}(\mathbf{f})$ denote the surrogate objective in \eqref{eq:P2_SCA}.
If $\mathbf{f}^{(i)}$ is feasible and the update $\mathbf{f}^{(i+1)}$ obtained from \eqref{eq:P2_SCA} satisfies $Q^{(i)}(\mathbf{f}^{(i+1)})\le Q^{(i)}(\mathbf{f}^{(i)})$, then
\begin{equation}
F(\mathbf{f}^{(i+1)}) \le Q^{(i)}(\mathbf{f}^{(i+1)}) \le Q^{(i)}(\mathbf{f}^{(i)})=F(\mathbf{f}^{(i)}).
\label{eq:sca_nonincrease}
\end{equation}
Thus, the objective of \eqref{eq:P2} is non-increasing in SCA iterations.

\subsubsection{Low-Dimensional Solution Structure}\label{sec:low_dim_subspace}
Both the BS-side beamformer design problem P2 in \eqref{eq:P2} and each SCA subproblem P2-SCA in \eqref{eq:P2_SCA} have an inherent low-dimensional solution structure as shown below. The optimizer can be restricted to the subspace spanned by the channels.

\begin{lemma}[Low-dimensional solution structure]\label{lem:channel_subspace}
Denote $\mathbf{H}\triangleq [\mathbf{h}_1,\dots,\mathbf{h}_K]\in\mathbb{C}^{N_\text{t}\times K}$ and define the channel subspace using $\mathcal{S}\triangleq \mathrm{col}(\mathbf{H})$.
If problem~\eqref{eq:P2} is feasible, its minimum value is attained, and at least one minimizer lies in $\mathcal{S}$.
Moreover, for any SCA iteration $i$, if the convex subproblem~\eqref{eq:P2_SCA} is feasible, it has a minimizer in $\mathcal{S}$.
\end{lemma}

\begin{proof}
When problem \eqref{eq:P2} is feasible, constraint C2 bounds the norm of $\mathbf{f}$, and all constraints are closed. Hence, the feasible set is nonempty and compact. Constraint $\overline{\text{C1}}$ also gives $|\mathbf{h}_k^{\mathsf H}\mathbf{f}|^2\ge u_k^2/\beta_{\max}^2>0$, so the reciprocal terms in the objective are finite and continuous on the feasible set. Therefore, the minimum is attained by the extreme value theorem. The same compactness and continuity argument applies to any feasible SCA subproblem. 

Decompose an arbitrary $\mathbf{f}$ as $\mathbf{f}=\mathbf{f}_{\parallel}+\mathbf{f}_{\perp}$, where $\mathbf{f}_{\parallel}\in\mathcal{S}$ and ${\mathbf{f}_{\perp}\in\mathcal{S}^{\perp}}$. Here, $\mathcal{S}^{\perp}$ denotes the orthogonal complement of $\mathcal{S}$.
Since each $\mathbf{h}_k\in\mathcal{S}$, we have $\mathbf{h}_k^{\mathsf{H}}\mathbf{f}=\mathbf{h}_k^{\mathsf{H}}\mathbf{f}_{\parallel}$ for all $k$.
Hence, all terms depending on $\{\mathbf{h}_k^{\mathsf{H}}\mathbf{f}\}_{k=1}^K$ in the objective and constraints of \eqref{eq:P2} and \eqref{eq:P2_SCA} depend only on $\mathbf{f}_{\parallel}$.
On the other hand, $\|\mathbf{f}\|_2^2=\|\mathbf{f}_{\parallel}\|_2^2+\|\mathbf{f}_{\perp}\|_2^2\ge\|\mathbf{f}_{\parallel}\|_2^2$.
Therefore, replacing $\mathbf{f}$ by $\mathbf{f}_{\parallel}$ preserves feasibility and does not increase the objective value. Applying this projection to a minimizer of \eqref{eq:P2} yields a minimizer in $\mathcal{S}$, and applying it to a minimizer of \eqref{eq:P2_SCA} yields an SCA-subproblem minimizer in $\mathcal{S}$.
\end{proof}

Lemma~\ref{lem:channel_subspace} enables an exact dimension reduction.
Define $r\triangleq \mathrm{rank}(\mathbf{H})\le K$, and let $\mathbf{U}\in\mathbb{C}^{N_\text{t}\times r}$ have orthonormal columns spanning $\mathcal{S}$.
Restricting SCA to this subspace preserves the attainable optimum of each subproblem while significantly reducing the problem dimension from $N_\text{t}$ to $r\le K$. This process requires a one-time matrix factorization.
In implementation, we obtain $\mathbf{U}$ using QR factorization $\mathbf{H}=\mathbf{U}\mathbf{R}$, and then parameterize the beamformer as
\begin{equation}
  \mathbf{f}=\mathbf{U}\mathbf{b},\qquad \mathbf{b}\in\mathbb{C}^{r}. \label{eq:w_subspace_param}
\end{equation}
Define reduced-dimension channels $\tilde{\mathbf{h}}_k\triangleq \mathbf{U}^{\mathsf{H}}\mathbf{h}_k\in\mathbb{C}^{r}$, then
\begin{equation}
  \mathbf{h}_k^{\mathsf{H}}\mathbf{f}=\tilde{\mathbf{h}}_k^{\mathsf{H}}\mathbf{b},\qquad \|\mathbf{f}\|_2^2=\|\mathbf{b}\|_2^2. \label{eq:reduced_channels}
\end{equation}
Under this parameterization, the affine lower bound in \eqref{eq:uk_affine_lb} is
\begin{equation}
\underline{u}_k^{(i)}(\mathbf{b})
=2\Re\!\{(a_k^{(i)})^\ast\tilde{\mathbf{h}}_k^{\mathsf{H}}\mathbf{b}\}-|a_k^{(i)}|^2,
\qquad a_k^{(i)}=\tilde{\mathbf{h}}_k^{\mathsf{H}}\mathbf{b}^{(i)}.
\label{eq:reduced_affine_lb}
\end{equation}
We solve problem P2-SCA in \eqref{eq:P2_SCA} using $\mathbf{b}$ (dimension $r\le K$) instead of $\mathbf{f}$ (dimension $N_\text{t}$),
and recover $\mathbf{f}=\mathbf{U}\mathbf{b}$ afterwards.

\subsubsection{Beamformer Initialization}
The SCA algorithm requires a feasible starting point. We adopt a maximum ratio transmission (MRT)-like beamformer in the reduced-dimensional subspace as a simple initialization. Specifically, we first obtain a reference beamforming direction $\mathbf{b}_{\mathrm{ref}} \triangleq \sum_{k=1}^{K}\tilde{\mathbf{h}}_k$. We then scale it to satisfy the BS power constraint C2 via
\begin{equation}
  \mathbf{b}^{(0)} \leftarrow \sqrt{P_{w,\max}/P_{w,0}}\,\frac{\mathbf{b}_{\mathrm{ref}}}{\|\mathbf{b}_{\mathrm{ref}}\|_2}. \label{MRT-like}
\end{equation}
If the result violates the mixer operating region constraint C3, we scale it down once more such that $\max_k |\tilde{\mathbf{h}}_k^{\mathsf{H}}\mathbf{b}^{(0)}| \le a_{\max}$. Before running the SCA, we also check a simple feasibility condition $u_k \le a_{\max}\beta_{\max}$ for all $k$. If this condition fails, the per-layer problem is considered infeasible.

\subsubsection{Summary of the Per-Layer Algorithm}\label{sec:per_layer_algorithm}
We summarize the per-layer joint beamforming and scaling design algorithm in Algorithm~\ref{alg:layer_sca}.
The algorithm first converts the root NMSE targets into the minimum end-to-end gains $\{u_k\}_{k=1}^{K}$ required by the clients. It then moves the beamforming update into the channel subspace, where the optimization variable has dimension $r\le K$ instead of $N_\text{t}$. In each SCA iteration, the nonconvex lower-gain constraint and reciprocal objective terms are handled through the affine lower bound in \eqref{eq:reduced_affine_lb}. Because this lower bound never exceeds the true squared gain and is exact at the current iterate, satisfying the SCA lower-gain constraint also satisfies the original lower-gain constraint. 
After the beamformer $\mathbf{f}$ is obtained, the client scaling coefficients $\{\beta_k\}_{k=1}^K$ are computed via \eqref{eq:beta_opt_form}.
The same algorithm is then applied to each layer $\ell$ with the corresponding parameters.

\begin{algorithm}[t]
\caption{Per-layer joint beamforming and scaling design}\label{alg:layer_sca}
\renewcommand{\algorithmicrepeat}{\textbf{repeat}}
\renewcommand{\algorithmicuntil}{\textbf{until}}
\begin{algorithmic}[1]
\STATE{\textbf{Input:}} Wireless channels $\{\mathbf{h}_k\}_{k=1}^K$, root NMSE targets $\{\epsilon_k\}_{k=1}^K$, system parameters $\kappa$, $c$, $P_{w,0}$, $P_{x,0}$, and $T_w$, power limits $P_{w,\max}$, $a_{\max}$ and $\beta_{\max}$, weighting parameter $\lambda\in[0,1]$, stopping tolerance $\tau>0$, and maximum number of iterations $I_{\max}$.
\STATE Compute gains $u_k \leftarrow \sqrt{1/(\kappa \epsilon_k^2)}$.
\STATE Form $\mathbf{H}\leftarrow [\mathbf{h}_1,\dots,\mathbf{h}_K]$, compute QR factorization $\mathbf{H}=\mathbf{U}\mathbf{R}$ to obtain orthonormal basis $\mathbf{U}\in\mathbb{C}^{N_\text{t}\times r}$ for $\mathrm{col}(\mathbf{H})$, and set reduced-dimensional channels as $\tilde{\mathbf{h}}_k\leftarrow \mathbf{U}^{\mathsf{H}}\mathbf{h}_k$.
\STATE Initialize $\mathbf{b}^{(0)}$ using the MRT-like beamformer in \eqref{MRT-like}.
\STATE Initialize the iteration index $i\leftarrow 0$.
\REPEAT
  \STATE Set $a_k^{(i)} \leftarrow \tilde{\mathbf{h}}_k^{\mathsf{H}}\mathbf{b}^{(i)}$ and form $\underline{u}_k^{(i)}(\mathbf{b})$ as \eqref{eq:reduced_affine_lb}, $\forall k \in \mathcal{K}$.
  \STATE Solve the reduced-dimensional SCA problem in \eqref{eq:P2_SCA} by replacing every occurrence of $\mathbf{h}_k^{\mathsf{H}}\mathbf{f}$ with $\tilde{\mathbf{h}}_k^{\mathsf{H}}\mathbf{b}$ and $\|\mathbf{f}\|_2^2$ with $\|\mathbf{b}\|_2^2$, to obtain $\mathbf{b}^{(i+1)}$.
  \STATE Set $i\leftarrow i+1$.
\UNTIL{$i\ge I_{\max}$ or $|F(\mathbf{U}\mathbf{b}^{(i)})-F(\mathbf{U}\mathbf{b}^{(i-1)})|\le\tau$}
\STATE Obtain the BS beamformer $\mathbf{f}\leftarrow \mathbf{U}\mathbf{b}^{(i)}$.
\STATE Obtain client scaling coefficients $\beta_k \leftarrow u_k/(\tilde{\mathbf{h}}_k^{\mathsf{H}}\mathbf{b}^{(i)})$.
\STATE \textbf{Output:} Beamformer $\mathbf{f}$, scaling coefficients $\{\beta_k\}_{k=1}^{K}$.
\end{algorithmic}
\end{algorithm}

\subsection{Complexity and Convergence}\label{sec:complexity_convergence}
As we apply the proposed per-layer algorithm to all $L$ layers, the overall complexity is $L$ times the per-layer complexity.
For a given layer, computing the required gains $\{u_k\}_{k=1}^{K}$ via $u_k=\sqrt{1/(\kappa \epsilon_k^2)}$ costs $\mathcal{O}(K)$ scalar operations. Computing $\mathbf{U}$ via a QR factorization of $\mathbf{H}\in\mathbb{C}^{N_\text{t}\times K}$ costs $\mathcal{O}(N_\text{t}K^2)$ when $N_\text{t}\gg K$.
After this reduction, each SCA iteration solves a convex problem in only $r$ complex variables, i.e., $2r$ real variables. With $K$ second-order cone constraints in C3 and $K$ affine constraints in $\widetilde{\overline{\text{C1}}}$, a generic interior-point method scales as approximately $\mathcal{O}\!\big((2r+K)^3\big)$ per iteration. Since $r\le K$, the per-iteration cost is effectively $\mathcal{O}(K^3)$ and is independent of the number of antennas $N_\text{t}$. With at most $I_{\max}$ SCA iterations per layer, the worst-case complexity is $\mathcal{O}\!\Big(L\big(N_\text{t}K^2 + K + I_{\max}(2r+K)^3\big)\Big)$. If the BS employs a large-scale antenna array, i.e., $r\le K\ll N_\text{t}$, this simplifies to $\mathcal{O}\!\Big(L\big(N_\text{t}K^2 + I_{\max}K^3\big)\Big)$. 

The non-increasing relationship in \eqref{eq:sca_nonincrease} suggests that each SCA step does not increase $F(\mathbf{f}^{(i)})$. If the initialization is feasible and each SCA update satisfies $Q^{(i)}(\mathbf{f}^{(i+1)})\le Q^{(i)}(\mathbf{f}^{(i)})$, then $\{F(\mathbf{f}^{(i)})\}$ is non-increasing. Since $F(\mathbf{f})$ is bounded below on the feasible set, this sequence converges to a finite limit~\cite{sun2017majorization}. Assuming that each SCA subproblem is solved exactly and that the linearized active inequality constraints of problem P2 in \eqref{eq:P2} admit a strictly feasible descent direction at each limit point, then each limit point of $\{\mathbf{f}^{(i)}\}$ is a stationary point of problem P2~\cite{scutari2017parallel}. Finally, since $\{\beta_k\}_{k=1}^K$ are recovered by the optimal solution in Lemma~\ref{lem:beta_structure}, the pair $(\mathbf{f},\{\beta_k\}_{k=1}^K)$ is a stationary solution of problem P1 in \eqref{eq:Pl_layer}.

\setlength{\abovedisplayskip}{3pt plus 1pt minus 1pt}
\setlength{\belowdisplayskip}{3pt plus 1pt minus 1pt}
\setlength{\belowdisplayshortskip}{3pt plus 1pt minus 1pt}
\section{Uniform- and Mixed-Precision Inference}
\label{sec:uniform_mixed_precision}

The aforementioned Algorithm~\ref{alg:layer_sca} can solve problem P1 in \eqref{eq:Pl_layer} to obtain the corresponding physical layer design for any prescribed root NMSE targets $\{\epsilon_k^{(\ell)}\}_{k,\ell}$. This means analog RF computing can support not only uniform-precision but also mixed-precision\footnote{In this paper, \textit{accuracy} and \textit{precision} are used interchangeably and carry the same meaning. We use precision here because \textit{mixed precision inference} is a standard terminology in AI inference on digital computing platforms.} NN inference. To distinguish the precision targets of different NN layers, we include the layer index $\ell$ throughout this section.

Uniform-precision inference is the default setting in most NN deployments. For example, when an NN is executed on a GPU, all layers are often computed using the same numerical precision, e.g., 16 or 32 bit. However, this can be improved because different NN layers often contribute differently to the final inference performance. This motivates mixed-precision inference, in which the most important layers are assigned tighter computing accuracy targets, while the less sensitive layers are assigned looser ones. In this way, better AI inference performance can be achieved with lower energy consumption. This has received attention for edge AI inference on digital computing platforms following~\cite{wang2019haq}.

Mixed-precision inference is also promising for AI inference based on analog computing, including the proposed analog RF computing. Since mixed-precision inference has been widely studied for digital computing, the most direct way to use this idea here is to inherit a digital mixed-precision allocation, and map the bit-widths for each NN layer to layerwise analog root NMSE targets through the equivalent accuracy model in~\cite{garg2021dynamic}, and then solve problem P1 layer by layer using Algorithm~\ref{alg:layer_sca} to realize the mixed-precision root NMSE targets through physical layer design. In most cases, this works well as Algorithm~\ref{alg:layer_sca} supports any accuracy targets $\{\epsilon_k^{(\ell)}\}$.

However, directly using digital mixed-precision allocations also has a limitation. Digital hardware supports only a \textit{discrete} and \textit{coarse} set of integer bit-widths, but analog RF computing supports a \textit{continuous} and \textit{fine-grained} accuracy control via the end-to-end gain $|a_k^{(\ell)}\beta_k^{(\ell)}|$, as can be seen in \eqref{eq:nmse_abeta}. This motivates the following question. Given a candidate layerwise precision allocation, are there any better root NMSE targets for analog computing that achieve a lower inference loss under the same or lower energy consumption? A general problem is
\begin{equation}
\begin{aligned}
\underset{\{\epsilon_k^{(\ell)}\}_{k,\ell}}{\text{minimize}}~~
& \mathcal{L}_{\mathrm{loss}}\big(\{\epsilon_k^{(\ell)}\}_{k,\ell}\big) \\
\text{subject~to}~~
& \sum_{\ell=1}^{L} E_{\lambda,\ell}^{\star}
\big(\{\epsilon_k^{(\ell)}\}_{k=1}^{K}\big) \le \Gamma_{\lambda},
\end{aligned}
\label{eq:mixed_precision_general}
\end{equation}
where $\mathcal{L}_{\mathrm{loss}}(\cdot)$ denotes the NN loss after injecting analog computing noise according to the prescribed root NMSE targets to each layer, $\Gamma_{\lambda}$ denotes the weighted energy budget associated with the same BS-client weighting factor $\lambda$ as in problem P1, and $E_{\lambda,\ell}^{\star}(\{\epsilon_k^{(\ell)}\}_{k=1}^{K})$ denotes the minimum weighted physical-layer energy of layer $\ell$, obtained by solving problem P1 under the root NMSE targets. Given a reference target profile $\{\bar{\epsilon}_k^{(\ell)}\}_{k,\ell}$, such as one inherited from digital mixed-precision inference, one may set $\Gamma_{\lambda}=\sum_{\ell=1}^{L}E_{\lambda,\ell}^{\star}(\{\bar{\epsilon}_k^{(\ell)}\}_{k=1}^{K})$. Then, problem \eqref{eq:mixed_precision_general} optimizes for another target profile that yields a lower NN loss while requiring no more weighted energy consumption compared to the reference profile. Problem~\eqref{eq:mixed_precision_general} is conceptually simple, however, it is a bilevel nonconvex optimization problem that is expensive to solve directly. Each evaluation of the energy constraint requires solving the \mbox{nonconvex} per-layer problem P1 in Section~\ref{sec:solver}, while the loss function couples all layers through the end-to-end NN inference process. Solving the problem would repeatedly invoke Algorithm~\ref{alg:layer_sca} across many candidate target profiles and channel realizations, which is in general computationally prohibitive.

Fortunately, the problem formulated in \eqref{eq:mixed_precision_general} is much simpler in some special cases of practical interest. For edge AI inference, we mainly focus on the energy consumption at the clients because edge devices are usually resource-constrained. This corresponds to the case where the weighting parameter $\lambda \approx 0$. To reveal this structure, we set $\lambda=0$ and temporarily ignore the BS power constraint C2 in problem P2. As only the optimal objective value $E_{\lambda,\ell}^{\star}
\big(\{\epsilon_k^{(\ell)}\}_{k=1}^{K}\big)$ is required in the constraint of problem \eqref{eq:mixed_precision_general}, we can replace the beamformer $\mathbf{f}^{(\ell)}$ by the effective gains $a_k^{(\ell)}=\mathbf{h}_k^{\mathsf H}\mathbf{f}^{(\ell)}$ in this simplified value calculation. The resulting problem decouples across clients and becomes
\begin{equation}
\begin{aligned}
\underset{\{|a_k^{(\ell)}|^2\}}{\text{minimize}}~~
& \sum_{k=1}^{K}\frac{c^{(\ell)} (u_k^{(\ell)})^2}{|a_k^{(\ell)}|^2} \\
\text{subject~to}~~
& \frac{(u_k^{(\ell)})^2}{\beta_{\max}^2} \le |a_k^{(\ell)}|^2 \le a_{\max}^2, ~~\forall k \in \mathcal{K}.
\end{aligned}
\label{eq:dp_lambda0_special}
\end{equation}
Because each term in the objective is decreasing in $|a_k^{(\ell)}|^2$, the optimal solution is $|a_k^{(\ell)}|^2=a_{\max}^2,\forall k \in \mathcal{K}$, and the optimal objective value is given by
\begin{equation}
\begin{aligned}
E_{0,\ell}^{\star}(\{\epsilon_k^{(\ell)}\}_{k=1}^{K})
&=\frac{1}{a_{\max}^2}\sum_{k=1}^{K} c^{(\ell)} (u_k^{(\ell)})^2 \\
&=\frac{1}{a_{\max}^2\kappa}\sum_{k=1}^{K}\frac{c^{(\ell)}}{(\epsilon_k^{(\ell)})^2},
\end{aligned}
\label{eq:dp_lambda0_exact_value}
\end{equation}
which reveals the inverse-square dependence of the client-side energy on the target root NMSE.

In the layerwise mixed-precision inference considered here, we further impose $\epsilon_k^{(\ell)}=\epsilon^{(\ell)},\forall k \in \mathcal{K}$. This follows the common layerwise mixed-precision convention, where all clients computing the same layer of a NN use the same accuracy target, while different layers can use different targets. Then \eqref{eq:dp_lambda0_exact_value} leads to the compact problem
\begin{equation}
\begin{aligned}
\underset{\{\epsilon^{(\ell)}\}_{\ell=1}^{L}}{\text{minimize}}~~
& \mathcal{L}_{\mathrm{loss}}\big(\{\epsilon^{(\ell)}\}_{\ell=1}^L\big) \\
\text{subject~to}~~
& \sum_{\ell=1}^{L} \frac{\omega_{\ell}}{(\epsilon^{(\ell)})^2} \le \Gamma_{0},
\end{aligned}
\label{eq:dp_lambda0_layerwise}
\end{equation}
where $\Gamma_{0}$ denotes the energy budget for $\lambda=0$, and $\omega_{\ell}
\triangleq
\frac{Kc^{(\ell)}}{a_{\max}^2\kappa}$. Problem \eqref{eq:dp_lambda0_layerwise} is a specialization of problem \eqref{eq:mixed_precision_general} in which the client-side energy value can be written explicitly as inversely proportional to the NMSE targets. As the constraint is greatly simplified, problem \eqref{eq:dp_lambda0_layerwise} can be solved by stochastic gradient descent. Specifically, we define $\gamma^{(\ell)} \triangleq \frac{1}{(\epsilon^{(\ell)})^2},\,\forall \ell \in \mathcal{L}$ such that the budget constraint becomes $\sum_{\ell=1}^{L} \omega_{\ell}\gamma^{(\ell)} \le \Gamma_0$. We then introduce a set of unconstrained variables $\{z^{(\ell)}\}_{\ell=1}^{L}$ and map them to budget shares via a softmax function, i.e., 
\begin{equation}
\pi^{(\ell)}
=
\frac{\exp(z^{(\ell)})}{\sum_{j=1}^{L}\exp(z^{(j)})},
\qquad
\sum_{\ell=1}^{L}\pi^{(\ell)}=1,
\label{eq:softmax_mixed}
\end{equation}
which leads to the parameterization
\begin{equation}
\omega_{\ell}\gamma^{(\ell)}
=
\omega_{\ell}\gamma_{\min}
+
\Big(\Gamma_0-\sum_{j=1}^{L}\omega_j\gamma_{\min}\Big)\pi^{(\ell)},
\qquad \forall \ell \in \mathcal{L},
\label{eq:budget_param_mixed}
\end{equation}
in which $\gamma_{\min}>0$ is a small precision floor. In \eqref{eq:budget_param_mixed}, the term $\omega_{\ell}\gamma_{\min}$ reserves energy for a minimum precision target for layer $\ell$, while $\pi^{(\ell)}$ allocates the remaining budget $\Gamma_0-\sum_{j=1}^{L}\omega_j\gamma_{\min}$. Summing \eqref{eq:budget_param_mixed} over $\ell$ gives $\sum_{\ell=1}^{L}\omega_{\ell}\gamma^{(\ell)}=\Gamma_0$, so the budget constraint is enforced by construction. The resulting $\epsilon^{(\ell)}=(\gamma^{(\ell)})^{-1/2}$ are used to inject layerwise Gaussian noise in end-to-end NN inference and compute a mini-batch estimate of $\mathcal{L}_{\mathrm{loss}}$. The NN weights are frozen and set as gradient detached in the process, and only variables $\{z^{(\ell)}\}_{\ell=1}^{L}$ are updated by stochastic gradient descent with the Adam optimizer \cite{kingma2015adam}, i.e., 
\begin{equation}
z^{(\ell)} \leftarrow z^{(\ell)} - \eta\,\frac{\partial \mathcal{L}_{\mathrm{loss}}}{\partial z^{(\ell)}},
\qquad \forall \ell \in \mathcal{L},
\label{eq:sgd_mixed}
\end{equation}
where $\eta$ is the step size. 

To summarize, uniform-precision inference uses a predetermined root NMSE target $\epsilon_{k}^{(\ell)}=\epsilon, \forall {k} \in \mathcal{K},\forall \ell \in \mathcal{L}$ shared by all clients and all layers. For mixed-precision inference, one can either directly use the analog equivalent of a digital mixed-precision bit-width allocation \cite{garg2021dynamic}, or, in the practically relevant client-energy-focused scenario with $\lambda\approx0$, solve problem \eqref{eq:dp_lambda0_layerwise} by using the stochastic gradient algorithm to obtain layerwise root NMSE targets. Once the root NMSE targets are determined, Algorithm~\ref{alg:layer_sca} is applied to problem P1 for each layer to obtain the physical layer design parameters $\mathbf{f}^{(\ell)}$ and $\{\beta_k^{(\ell)}\}_{k=1}^{K}$.

\setlength{\abovedisplayskip}{1.5ex plus 4pt minus 2pt}
\setlength{\belowdisplayskip}{1.5ex plus 4pt minus 2pt}
\setlength{\belowdisplayshortskip}{1.5ex plus 4pt minus 2pt}
\setlength{\textfloatsep}{8pt plus 2pt minus 2pt}
\setlength{\dbltextfloatsep}{8pt plus 2pt minus 2pt}
\setcounter{dbltopnumber}{1}
\setcounter{topnumber}{1}

\section{Performance Evaluation}\label{sec:sim_results}
In this section, we evaluate the proposed physical layer design for analog RF computing-based edge AI inference over the wireless downlink. The simulations jointly account for the NN architecture, the mixer-based energy model, the waveform schedule, and wireless channels generated according to the 3GPP specifications. We first describe the NN, hardware parameters, and channel generation procedure. We then evaluate the convergence and complexity of Algorithm~\ref{alg:layer_sca}, the energy-accuracy behavior, the BS-client energy tradeoff, and the benefits of mixed-precision analog AI inference.

\begin{table}[t]
\caption{Structure of the CNN inspired by LeNet-5.}
\label{tab:lenet_architecture}
\centering
\scriptsize
\renewcommand{\arraystretch}{0.95}
\setlength{\tabcolsep}{5.8pt}
\begin{tabular}{l p{0.42\columnwidth} c c c c}
\toprule
$\ell$ & Operation & $M^{(\ell)}$ & $N^{(\ell)}$ & $P^{(\ell)}$ & MACs \\
\midrule
1 & $5\times5$ conv., 6 output maps, $28\times28$ positions (followed by pooling) & 6 & 25 & 784 & \minorchange{470.4k} \\
2 & $5\times5$ conv., 16 output maps, $10\times10$ positions (followed by pooling) & 16 & 150 & 100 & \minorchange{960.0k} \\
3 & fully connected, $400\rightarrow120$ & 120 & 400 & 1 & \minorchange{192.0k} \\
4 & fully connected, $120\rightarrow84$ & 84 & 120 & 1 & \minorchange{40.32k} \\
5 & fully connected, $84\rightarrow10$ & 10 & 84 & 1 & \minorchange{3.36k} \\
\bottomrule
\end{tabular}
\end{table}

\subsection{Simulation Settings}\label{sec:simulation_settings}
We use the MNIST handwritten-digit dataset and a convolutional neural network (CNN) inspired by LeNet-5~\cite{lecun1998gradient}. \minorchange{Similar to~\cite{wise}, we consider complex-valued NN parameters since the passive mixer naturally supports complex-valued MVMs.}\footnote{\minorchange{Real-valued NNs are a special case of complex-valued NNs. Our ongoing study will explore waveform designs that make better use of the passive mixer for these real-valued operations.}} The CNN has two convolution layers (each followed by pooling layers), and three fully connected layers. Since the pooling and nonlinear activation operations have much lower arithmetic cost than the linear layers, we apply analog RF computing to the five linear layers listed in Table~\ref{tab:lenet_architecture}. Each convolution layer is converted into an MVM by standard unfolding. The corresponding $(M^{(\ell)},N^{(\ell)},P^{(\ell)})$ in layer $\ell$ and the number of MAC operations are also reported in Table~\ref{tab:lenet_architecture}. Unless otherwise stated, the main inference experiments use $N_\text{t}=256$ BS antennas, $K=10$ clients, $3000$ randomly selected MNIST test images for accuracy evaluation, and another $1000$ images as the calibration dataset to estimate $\{\bar S_{\mathrm{ref}}^{(\ell)}\}$. By default, we consider uniform-precision inference, and $\epsilon$ denotes the common root NMSE target for all clients and layers. Mixed-precision inference results are presented in Section~\ref{sec:sim_dynamic_precision}.

\nocite{3gpp38901,nvidia_a100_datasheet,etsi_bs,etsi_ue}

\begin{table}[t]
\caption{List of key simulation parameters.}
\label{tab:constants}
\centering
\scriptsize
\renewcommand{\arraystretch}{0.95}
\setlength{\tabcolsep}{2.4pt}
\begin{tabular}{l p{0.52\columnwidth} l}
\toprule
Symbol & Description & Value \\
\midrule
$B$ & Occupied bandwidth & $25$~MHz \\
$L_{\mathrm{fft}}$ & FFT length & $4096$ \\
$f_w$ & Downlink carrier frequency & $2.5$~GHz \\
$d_{2\mathrm{D},k}$ & Horizontal BS-client distance & uniform in $[10,15]$~m \\
$h_{\mathrm{BS}},h_{\mathrm{client}}$ & BS and client heights & $8$~m, $1.5$~m \\
$G_{\mathrm{BS}},G_{\mathrm{client}}$ & BS and client antenna gains & $8$~dBi, $3$~dBi \\
$K_{\mathrm R}$ & Rician factor & $9$~dB \\
$T_0$ & Environmental temperature & $300$~K \\
$k_BT_0$ & Thermal-noise power spectral density & $-174$~dBm/Hz \\
$\vartheta^{(\ell)}$ & Guard factor in layer $\ell$ & $0.33$ \\
$M^{\prime(\ell)}$ & Size of row blocks in layer $\ell$ & $6$ \\
$\varpi^{(\ell)}$ & CP overhead in layer $\ell$ \mbox{\cite{wise}} & $0.125$ \\
$e_{\mathrm{adc}}$ & ADC energy per sample \mbox{\cite{wise}} & $1$~pJ/sample \\
$e_{\mathrm{dig}}$ & Digital decoding energy per MAC \mbox{\cite{wise}} & $1$~pJ/MAC \\
$e_{\mathrm{digital}}$ & Digital computing energy per MAC~\mbox{\cite{nvidia_a100_datasheet}} & $3$~pJ/MAC \\
$\rho_{\mathrm{radio}}$ & Client RF-chain efficiency \mbox{\cite{wise}} & $0.30$ \\
$\rho_{\mathrm{mixer}}$ & Mixer conversion coefficient \mbox{\cite{wise}} & $0.2512$ \\
$\rho_{\mathrm{nf}}$ & Effective receiver noise coefficient \mbox{\cite{wise}} & $0.2512$ \\
$P_{w,\max}$ & BS transmit power limit~\mbox{\cite{etsi_bs}} & $48$~dBm \\
$P_{x,\max}$ & Client transmit power limit~\mbox{\cite{etsi_ue}} & $23$~dBm \\
$P_{w,0}$ & BS-side reference power level & $0$~dBm \\
$P_{x,0}$ & Client-side reference power level & $0$~dBm \\
$P_{\mathrm{LO,th}}$ & LO-port saturation threshold~\mbox{\cite{zem4300_datasheet}} & $-3$~dBm \\
\bottomrule
\end{tabular}
\end{table}

\begin{figure}[t]
  \centering
  \subfloat[]{\includegraphics[width=0.85\columnwidth]{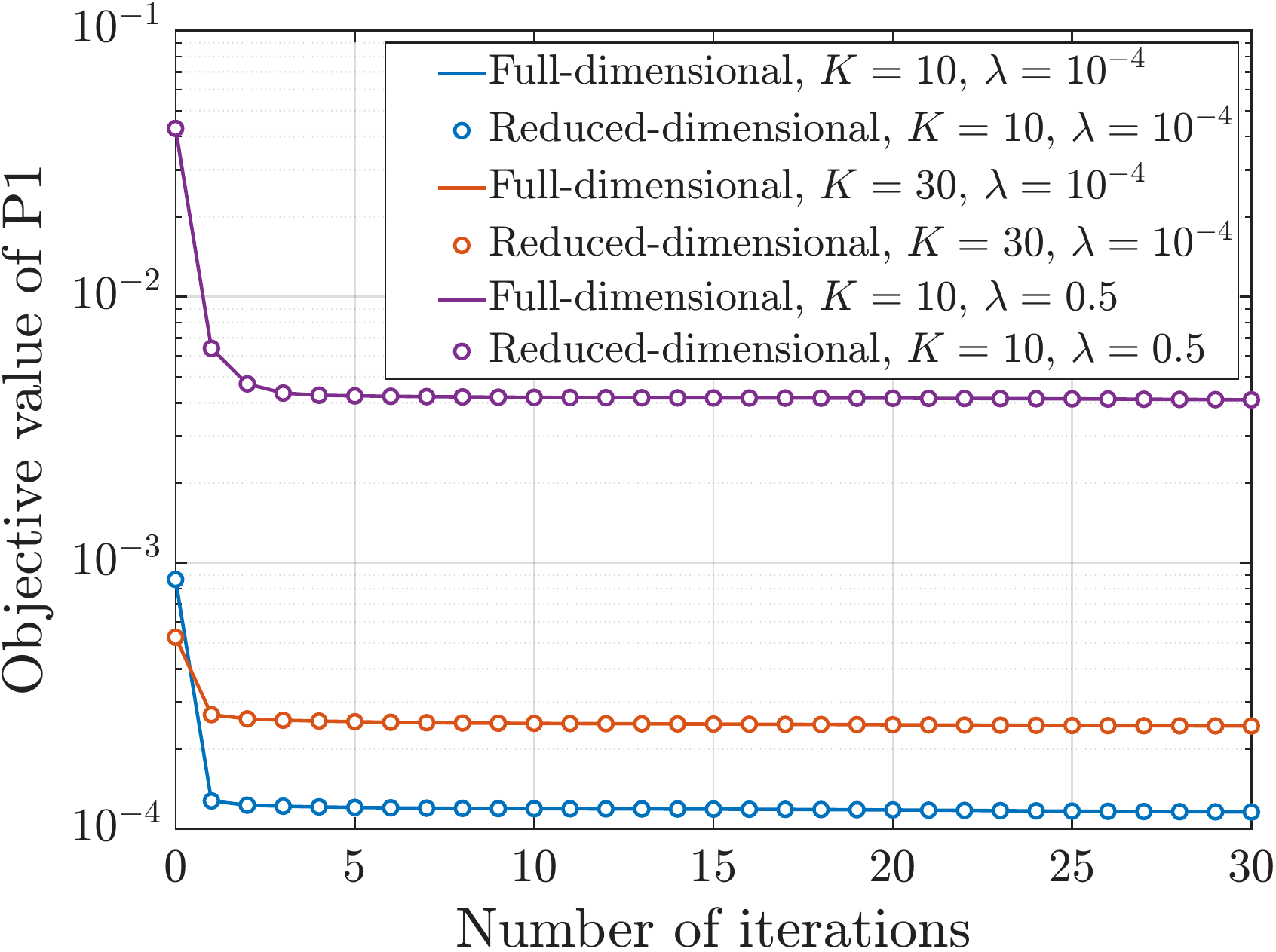}\label{fig:conv}}\par\smallskip
  \subfloat[]{\includegraphics[width=0.85\columnwidth]{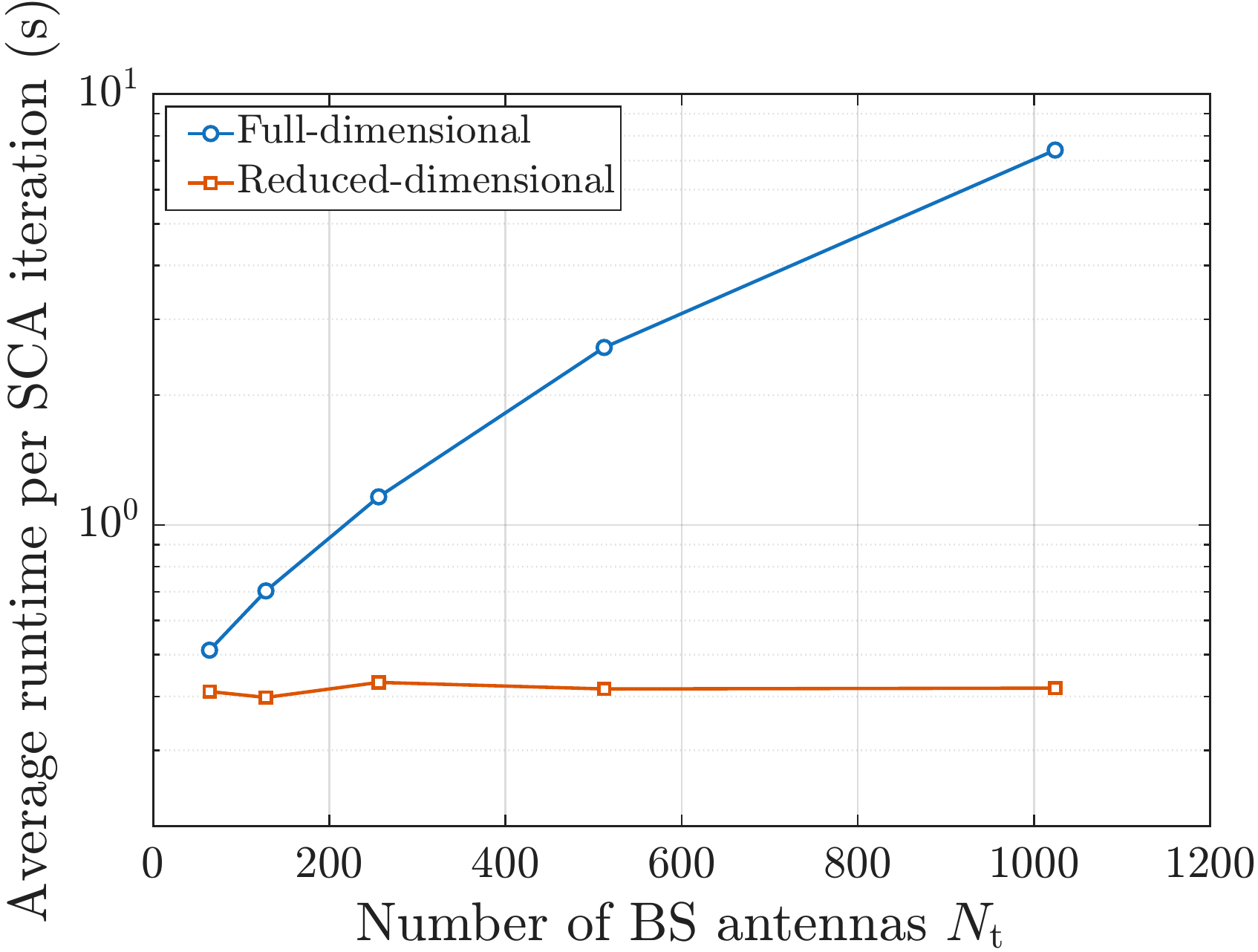}\label{fig:complexity}}
  \caption{Convergence and complexity of the proposed algorithm. }
  \label{fig:alg_results}
\end{figure}

The distance-dependent path loss, shadow fading, and line-of-sight (LoS) probability follow the 3GPP TR~38.901 Indoor Factory with Sparse clutter and High base station height (InF-SH) model~\cite{3gpp38901}. The BS uses a half-wavelength-spaced uniform planar array (UPA). For the default value $N_\text{t}=256$, the array size is $16\times16$. For client $k$, let $d_{2\mathrm{D},k}$ be the horizontal BS-client distance and $d_{3\mathrm{D},k}=\sqrt{d_{2\mathrm{D},k}^{2}+(h_{\mathrm{BS}}-h_{\mathrm{client}})^2}$ denote the three-dimensional distance. Unless otherwise stated, $d_{2\mathrm{D},k}$ is drawn uniformly from $[10,15]$~m. The horizontal azimuth angle $\phi_k$ is drawn uniformly from $[0,2\pi)$. The elevation angle is set by geometry as $\theta_k=\arctan((h_{\mathrm{client}}-h_{\mathrm{BS}})/d_{2\mathrm{D},k})$. The LoS probability is $p_{\mathrm{LoS},k}=\exp(-d_{2\mathrm{D},k}/k_{\text{InF-SH}})$. A Bernoulli trial with probability $p_{\mathrm{LoS},k}$ determines the LoS or non-LoS (NLoS) state of client $k$. For the InF-SH sparse-clutter setting used here, $k_{\text{InF-SH}}=-\frac{d_{\mathrm{clutter}}}{\ln(1-r_{\mathrm{clutter}})}\frac{h_{\mathrm{BS}}-h_{\mathrm{client}}}{h_c-h_{\mathrm{client}}}=582.6$~m, obtained by using $d_{\mathrm{clutter}}=10$~m, $r_{\mathrm{clutter}}=0.2$, $h_c=2$~m, $h_{\mathrm{BS}}=8$~m, and $h_{\mathrm{client}}=1.5$~m in the 3GPP InF-SH LoS-probability formula.

The deterministic path losses in dB are
\begin{equation*}
\begin{aligned}
\mathrm{PL}_{\mathrm{LoS},k}
&=31.84+21.5\log_{10}(d_{3\mathrm{D},k})
  +19\log_{10}(f_w),\\
\mathrm{PL}_{\mathrm{NLoS},k}
&=\max\{\mathrm{PL}_{\mathrm{LoS},k},
  32.4+23\log_{10}(d_{3\mathrm{D},k}) \\
&\qquad
  +20\log_{10}(f_w)\},
\end{aligned}
\end{equation*}
for the LoS and NLoS links, respectively, where $d_{3\mathrm D,k}$ is in meters and $f_w$ is in GHz. The coefficients $21.5$ and $23$ correspond to path-loss exponents $2.15$ and $2.3$, respectively. Let $\mathrm{PL}_{0,k}$ denote either $\mathrm{PL}_{\mathrm{LoS},k}$ or $\mathrm{PL}_{\mathrm{NLoS},k}$ according to the drawn LoS or NLoS state. The overall path loss is $\mathrm{PL}_{k}=\mathrm{PL}_{0,k}+\mathrm{SF}_{k}$, where $\mathrm{SF}_{k}$ is zero-mean log-normal shadow fading in dB, with standard deviation $4$~dB for LoS links and $5.9$~dB for NLoS links. The large-scale channel power gain after antenna gains is $\zeta_k=10^{-\left(\mathrm{PL}_k-G_{\mathrm{BS}}-G_{\mathrm{client}}\right)/10}$, and the channel vector is $\mathbf h_k=\sqrt{\zeta_k}\,\breve{\mathbf h}_k$.

For LoS clients, the small-scale fading vector is
\begin{equation}
\breve{\mathbf h}_k
=
\sqrt{\frac{K_{\mathrm R}}{K_{\mathrm R}+1}}e^{j\varphi_k}\mathbf a(\theta_k,\phi_k)
+
\sqrt{\frac{1}{K_{\mathrm R}+1}}\mathbf r_k,
\label{eq:sim_rician_channel}
\end{equation}
where $K_{\mathrm R}$ is the Rician factor in linear scale, $\mathbf r_k\sim\mathcal{CN}(\mathbf 0,\mathbf I)$ is the scattering component, and $\mathbf a(\theta_k,\phi_k)$ is the UPA steering vector for elevation angle $\theta_k$ and horizontal azimuth angle $\phi_k$. The scalar phase $\varphi_k$ is drawn independently from $[0,2\pi)$ and models the common propagation phase of the LoS path. This common phase is different from the antenna-dependent phase shifts captured by the steering vector. For NLoS clients, we use Rayleigh fading by setting $\breve{\mathbf h}_k=\mathbf r_k$. We adopt the remaining parameters in Table~\ref{tab:constants} in the simulations. The table cites the relevant sources for choosing the key parameters. The simulations are all carried out in MATLAB and the convex SCA subproblems are solved using CVX~\cite{cvx}.

\begin{figure}[t]
  \centering
  \subfloat[]{\includegraphics[width=0.8924\columnwidth]{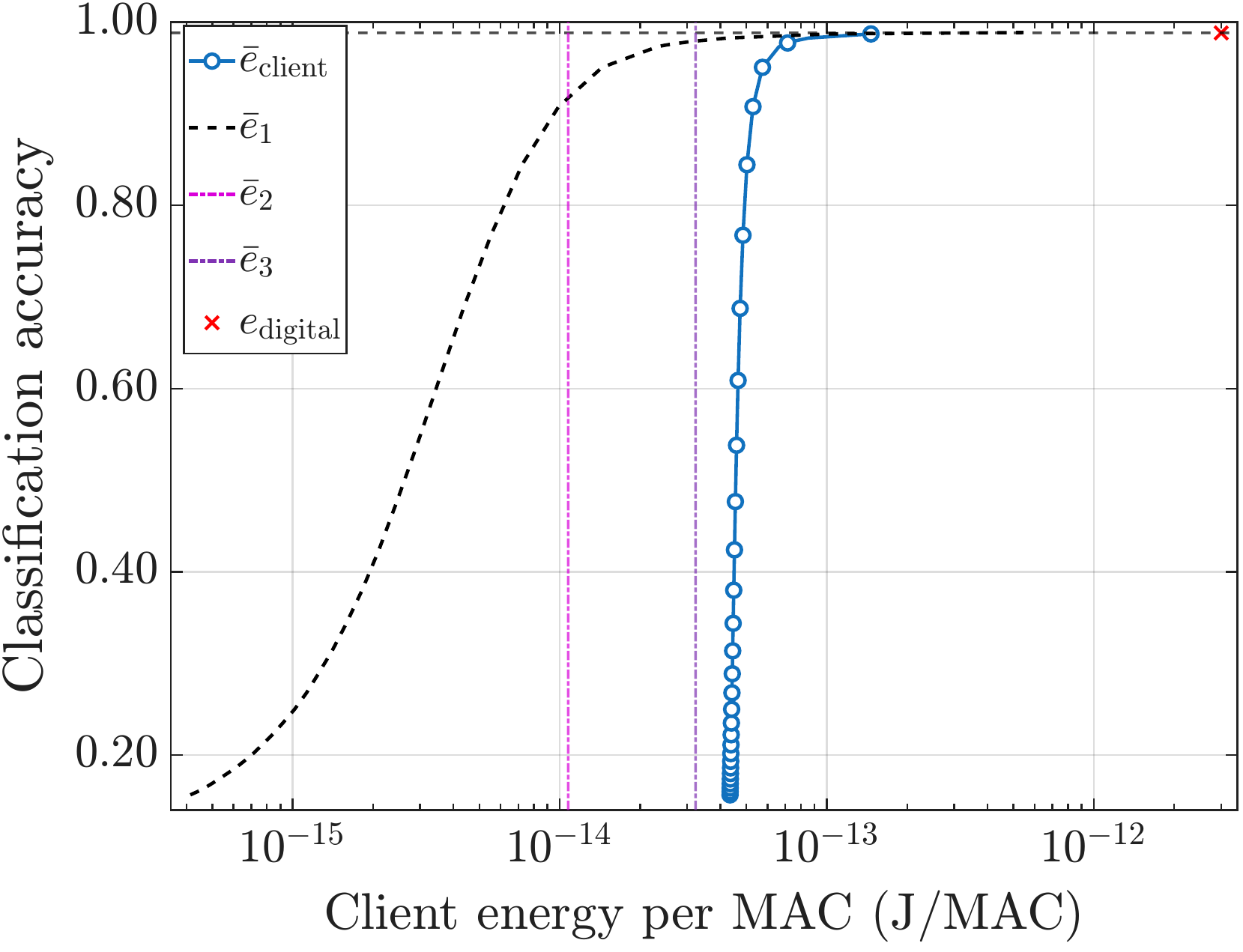}\label{fig:acc_energy}}\par\smallskip
  \subfloat[]{\includegraphics[width=0.8924\columnwidth]{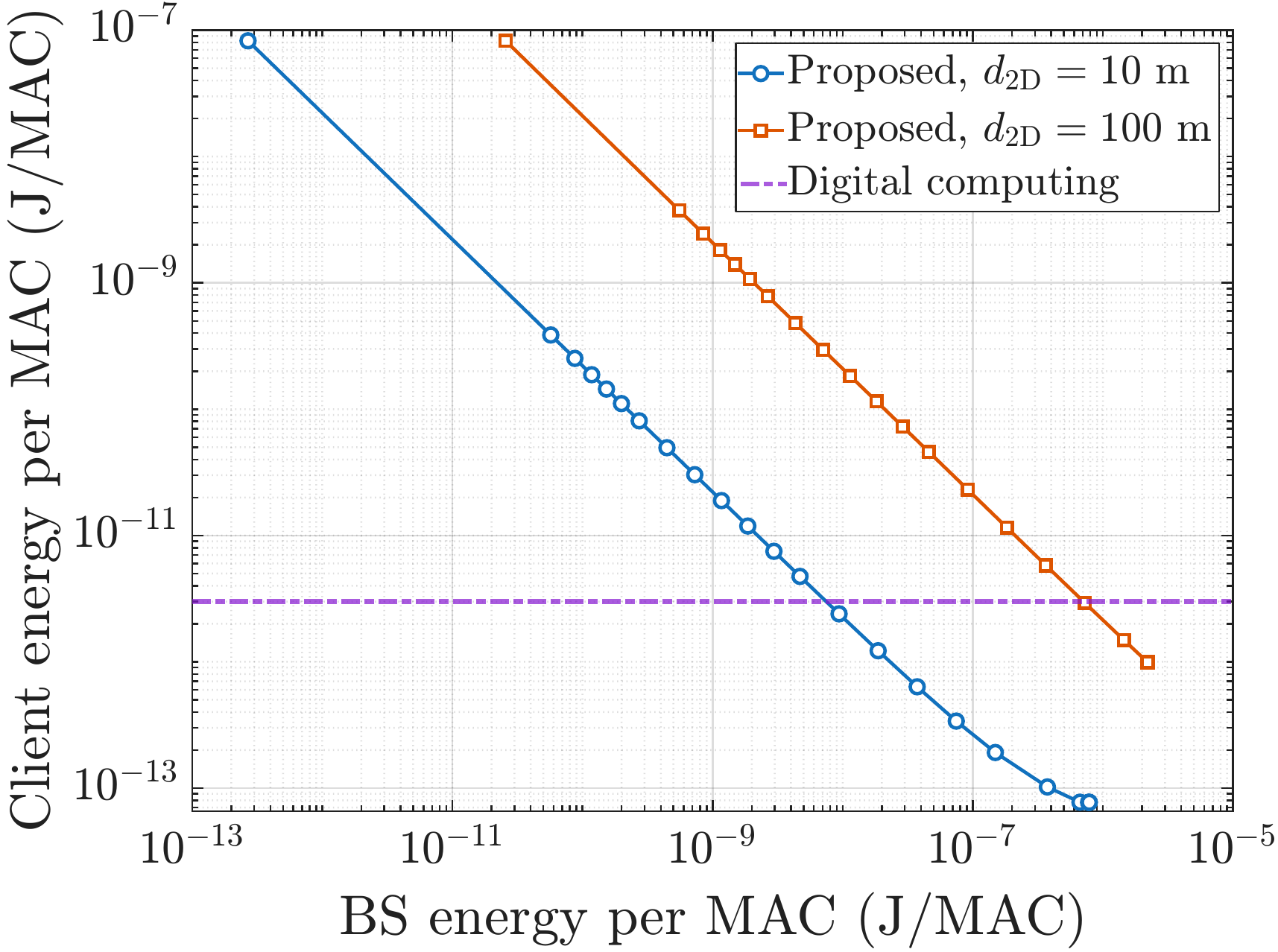}\label{fig:distance_tradeoff}}
  \caption{Energy consumption of analog RF computing-based edge inference.}
  \label{fig:energy_results}
\end{figure}

\begin{figure}[t]
  \centering
  \subfloat[]{\includegraphics[width=0.8924\columnwidth]{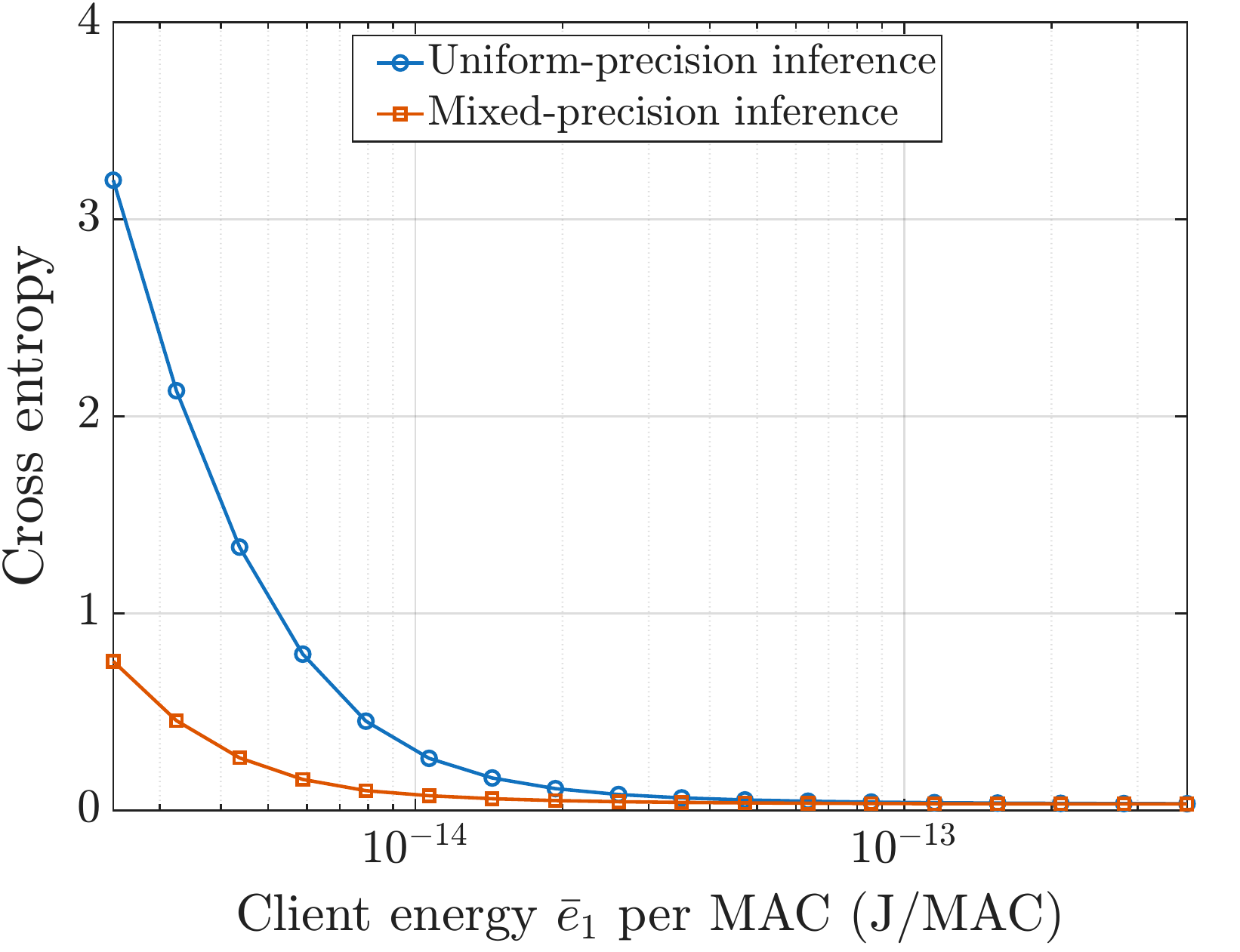}\label{fig:dyn_ce}}\par\smallskip
  \subfloat[]{\includegraphics[width=0.8924\columnwidth]{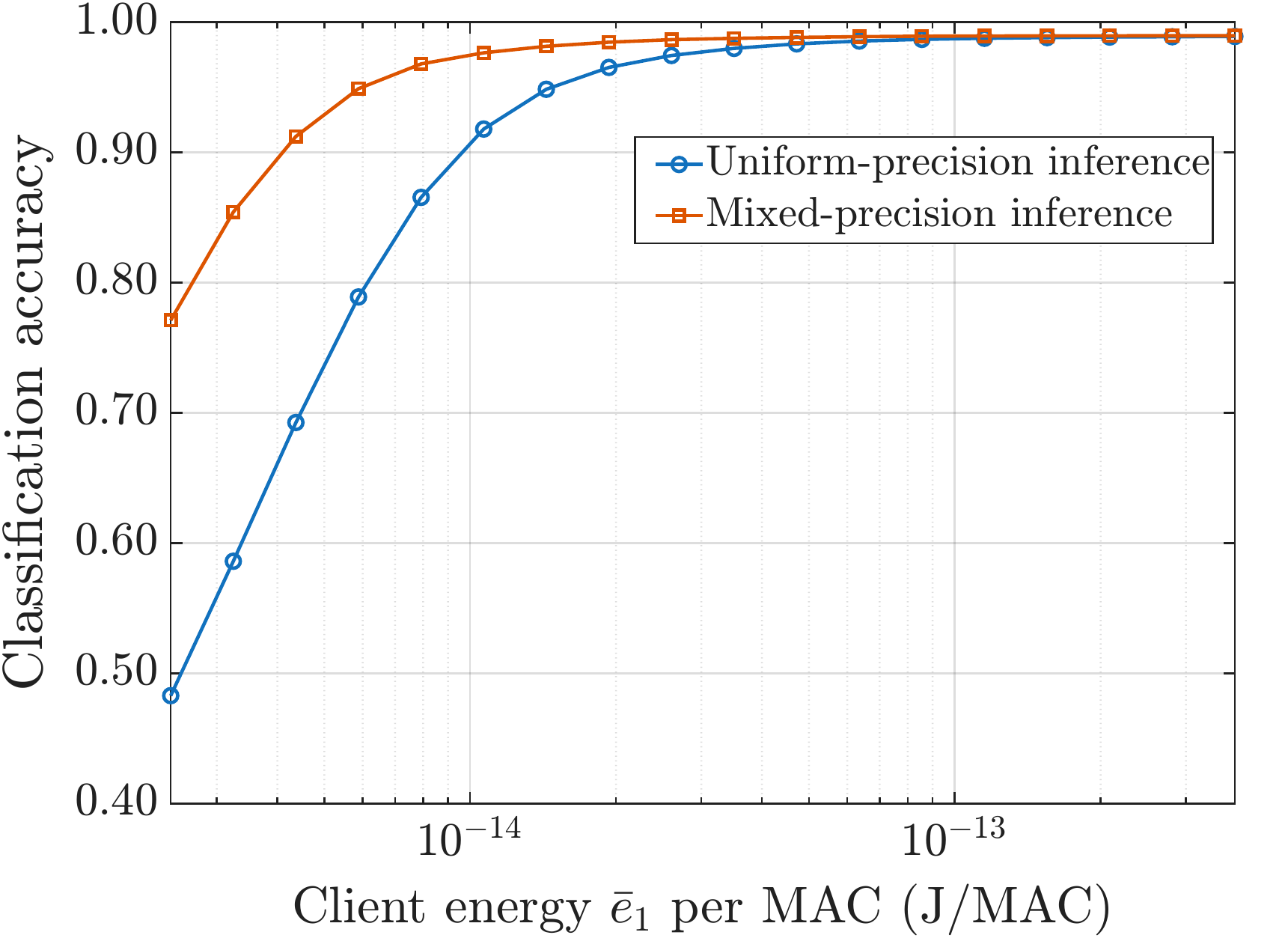}\label{fig:dyn_acc}}
  \caption{Comparison of uniform- and mixed-precision inference.}
  \label{fig:dyn_results}
\end{figure}

\subsection{Convergence and Complexity}\label{sec:sim_convergence_complexity}

Fig.~\ref{fig:alg_results} shows the convergence and complexity of Algorithm~\ref{alg:layer_sca}. Both subfigures use the first LeNet layer $\ell=1$ to isolate the behavior of the per-layer problem P1. In Fig.~\ref{fig:alg_results}\subref{fig:conv}, we consider uniform-precision inference with a common root NMSE target $\epsilon=0.1$, the maximum number of SCA iterations $I_{\max}$ is set to $30$, and we vary the number of clients $K$ and the weighting factor $\lambda$ as shown in the legend. In Fig.~\ref{fig:alg_results}\subref{fig:complexity}, we fix $K=10$, $\lambda=0.5$, and $\epsilon=0.1$, and sweep $N_\text{t}$ from $64$ to $1024$. The objective value of \eqref{eq:P2} decreases from its communication-oriented MRT-like initialization for all tested settings. The reduced-dimensional and full-dimensional curves overlap, which confirms that the exact channel-subspace restriction preserves the original solution trajectory. The runtime comparison further shows that the full-dimensional implementation grows rapidly with $N_\text{t}$, whereas the reduced-dimensional implementation stays nearly flat because its optimization dimension is controlled by the number of clients $K$. These results validate the convergence behavior of Algorithm~\ref{alg:layer_sca} and show the practical value of the reduced-dimensional algorithm.

\subsection{Energy-Accuracy and BS-Client Tradeoffs}\label{sec:sim_energy_accuracy}

Fig.~\ref{fig:energy_results} evaluates the end-to-end energy and accuracy tradeoffs under uniform precision. Fig.~\ref{fig:energy_results}\subref{fig:acc_energy} sweeps a common root NMSE target $\epsilon\in[0.05,2.0]$ across all clients and layers under $N_\text{t}=256$, $K=10$, and $\lambda=0$. The dashed curve isolates $\bar e_1$, the average waveform-generation energy per MAC, while the horizontal offset to the total curve accounts for the fixed costs $\bar e_2$ and $\bar e_3$. The classification accuracy remains close to digital computing over a broad low-energy range. Compared to the fixed digital reference $e_{\mathrm{digital}}=3$~pJ/MAC \cite{nvidia_a100_datasheet}, the proposed analog RF computing design achieves a much lower client-side inference energy in the high-accuracy region.

Fig.~\ref{fig:energy_results}\subref{fig:distance_tradeoff} shows the BS-client energy tradeoff at two propagation distances. We sweep $\lambda$ with $N_\text{t}=256$, $K=10$, and \minorchange{$\epsilon=0.2$, which corresponds to a classification accuracy of $98.04\%$.} The client-side energy decreases as the BS energy increases, because stronger BS-side beamforming gain reduces the required client-side scaling. The $10$~m frontier is shifted down and left relative to the $100$~m frontier, which confirms that stronger channels reduce the energy required at both the BS and the clients for the same computing accuracy. The horizontal dashed line denotes the client-side energy of a digital-computing reference, where the inference computation is carried out locally at the clients. At $100$~m and $\lambda=0$, analog RF computing achieves about a \minorchange{$3.02\times$} reduction in client-side energy compared with this digital reference. This result suggests that the service distance can be further extended when the channel is more favorable, the BS transmit-power budget is higher, or higher-gain antennas are being used.

\begin{figure}[t]
\centering
\includegraphics[width=0.88\columnwidth]{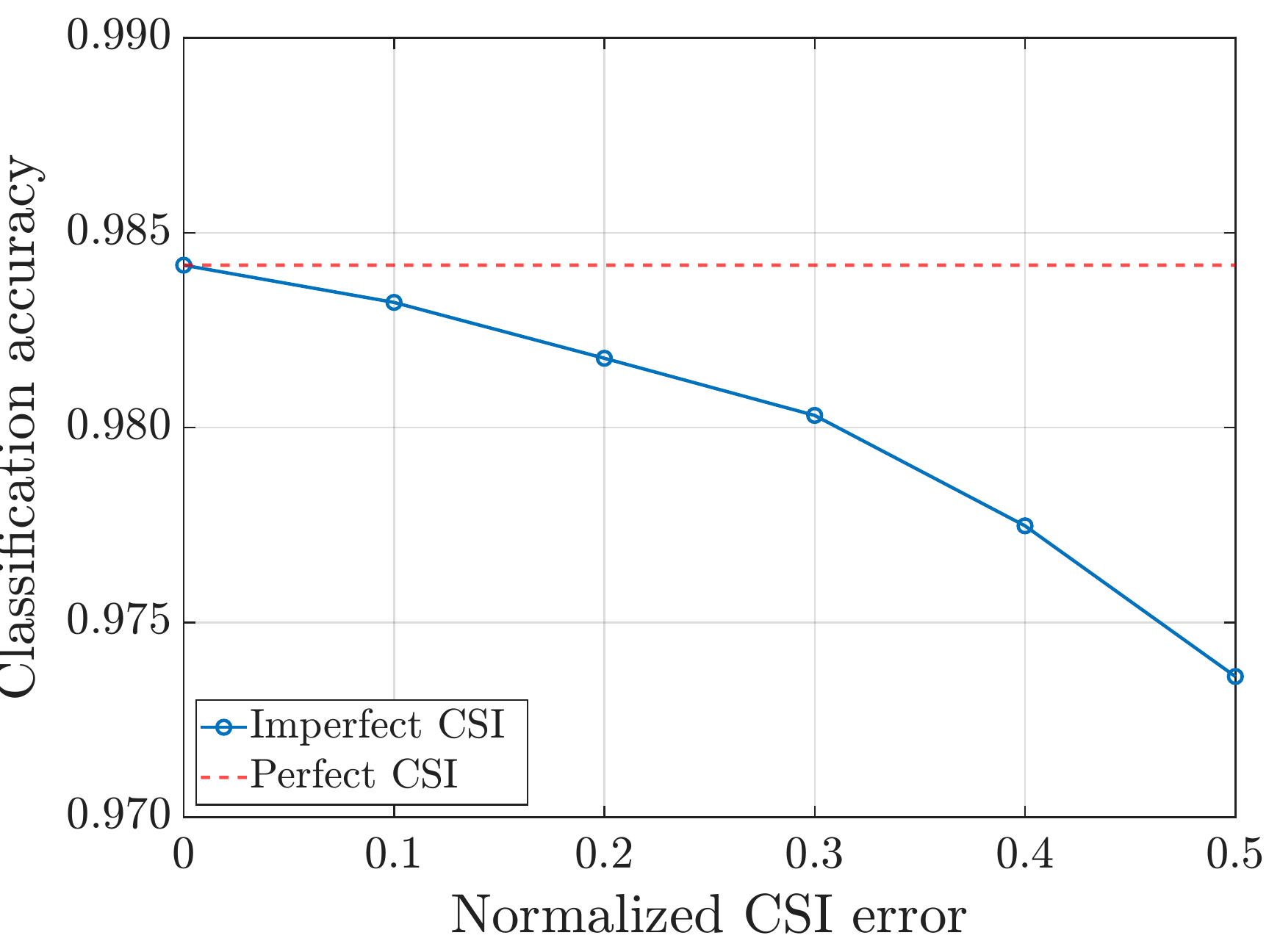}
\caption{Classification accuracy versus the normalized CSI error.}
\label{fig:robustness_csi}
\end{figure}

\begin{figure}[t]
\centering
\includegraphics[width=0.88\columnwidth]{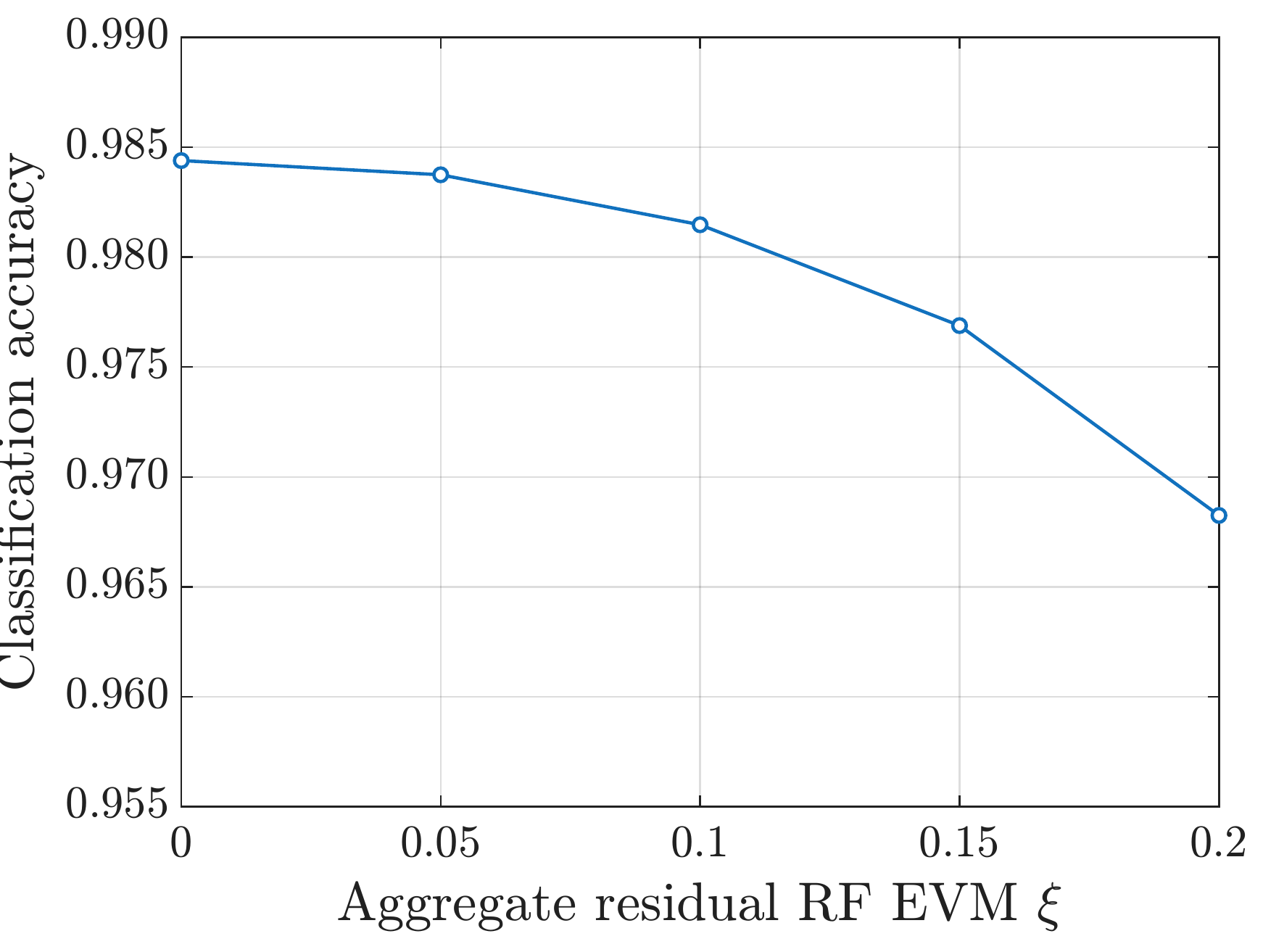}
\caption{Classification accuracy versus the aggregate residual RF EVM.}
\label{fig:robustness_rf}
\end{figure}

\subsection{Uniform- vs. Mixed-Precision Inference}\label{sec:sim_dynamic_precision}

We next compare uniform- and mixed-precision inference. In this experiment, we use a block-fading model in which all layers of one inference request experience the same channel realization, while different inference requests correspond to different channel realizations. This separates the offline target-allocation effect from the online channel-adaptive physical-layer design. The uniform-precision inference baseline sweeps a shared root NMSE target $\epsilon_{\mathrm{sh}}\in[0.06,0.95]$ over $18$ logarithmically spaced values for all layers, with $\lambda=0$. The mixed-precision root NMSE target profile is then optimized under the energy budget $\Gamma_0$ using the stochastic gradient descent algorithm proposed in Section \ref{sec:uniform_mixed_precision}.

Fig.~\ref{fig:dyn_results} reports the mixed-precision results for $N_\text{t}=256$ and $K=6$. Fig.~\ref{fig:dyn_results}\subref{fig:dyn_ce} plots the cross entropy, which is the loss function $\mathcal{L}_{\mathrm{loss}}(\cdot)$ we minimized in problem \eqref{eq:dp_lambda0_layerwise}, and Fig.~\ref{fig:dyn_results}\subref{fig:dyn_acc} plots the classification accuracy, both against the client-side waveform-generation energy $\bar e_1$ per MAC. The mixed-precision inference provides lower cross entropy and higher classification accuracy than the uniform-precision baseline in the low-to-moderate energy budget regime. The gain is most visible when the energy budget is limited, because mixed precision assigns tighter root NMSE targets to layers that are more sensitive to analog computing noise and relaxes the targets of less sensitive layers. As the energy budget increases, the performance gap gradually narrows. These results confirm that mixed-precision analog AI inference is an effective way to achieve the same inference performance with lower energy consumption, or in other words, to achieve better performance under the same energy consumption.

\subsection{Robustness to Imperfect CSI}\label{sec:sim_csi}
\setcounter{topnumber}{1}

Channel estimation errors can make the available channel state information (CSI) differ from the actual channel during deployment. In this subsection, we study the robustness of analog RF computing to imperfect CSI. Fig.~\ref{fig:robustness_csi} plots the classification accuracy against the normalized CSI error $\eta_{\mathrm{CSI}}=\mathbb E[\|\mathbf H_{\mathrm{true}}-\widehat{\mathbf H}\|_{\mathrm F}^2]/\mathbb E[\|\mathbf H_{\mathrm{true}}\|_{\mathrm F}^2]$, where $\mathbf H_{\mathrm{true}}$ and $\widehat{\mathbf H}$ are the actual and estimated channel matrices, respectively. In the simulations, we set $N_\mathrm t$, $K$, $\lambda$, and $\epsilon$ to $256$, $10$, $0.5$, and $0.18$, respectively.

Fig.~\ref{fig:robustness_csi} shows that the classification accuracy decreases as the normalized CSI error increases. The reduction becomes more pronounced at larger CSI errors, as the physical layer design is less well matched to the actual channel. When the normalized CSI error reaches a large value of $0.5$, the classification accuracy is still \minorchange{$97.36\%$}. Compared with \minorchange{$98.42\%$} for perfect CSI, the relative reduction in accuracy is only \minorchange{$1.07\%$}. This demonstrates the robustness of the proposed analog RF computing-based edge inference to imperfect CSI.

\vspace{\baselineskip}
\begin{figure*}[t]
\centering
\begin{minipage}[t]{\columnwidth}
\centering
\includegraphics[height=0.68\columnwidth]{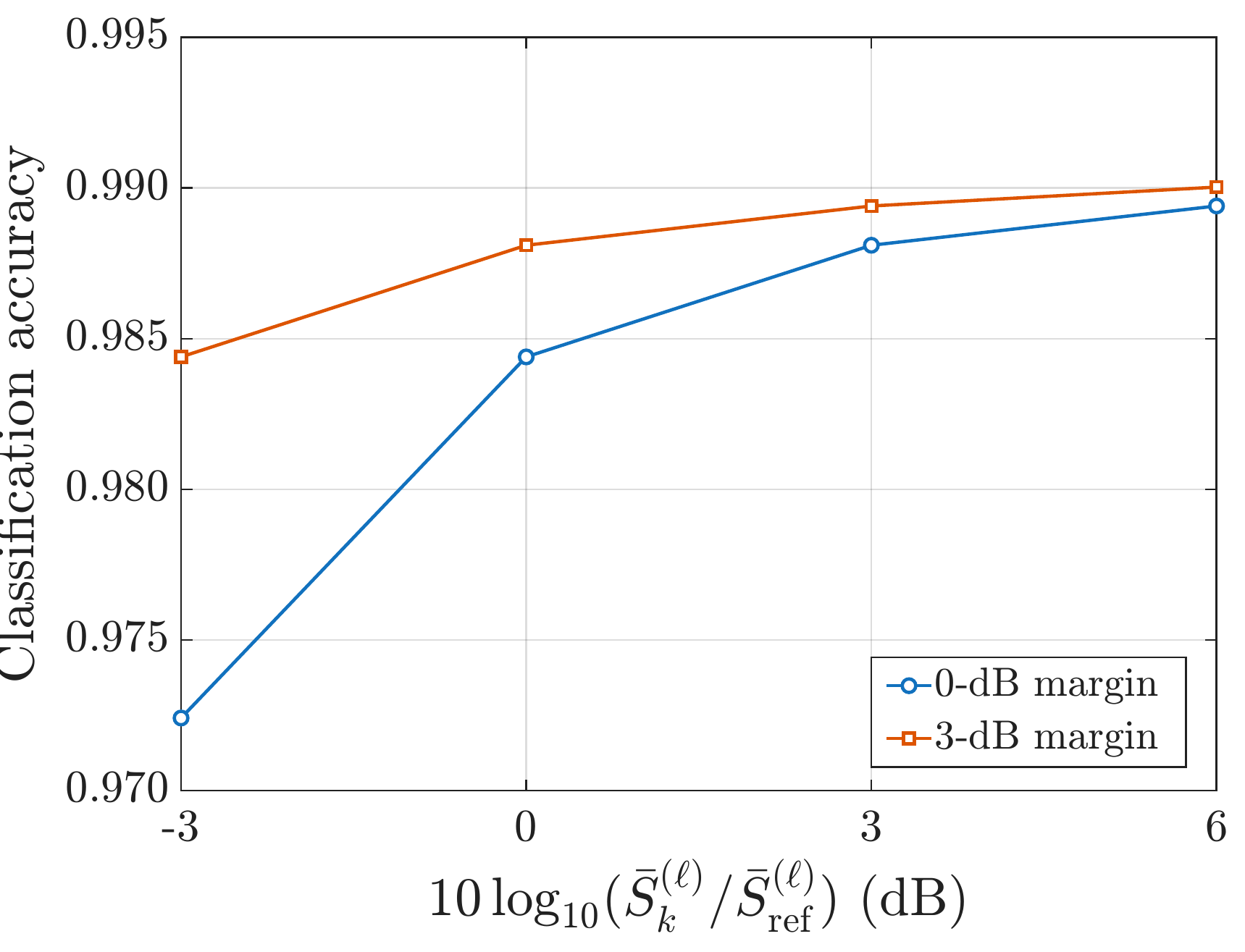}
\caption{Classification accuracy under reference power mismatch.}
\label{fig:calibration_mismatch}
\end{minipage}\hfill
\begin{minipage}[t]{\columnwidth}
\centering
\includegraphics[height=0.68\columnwidth]{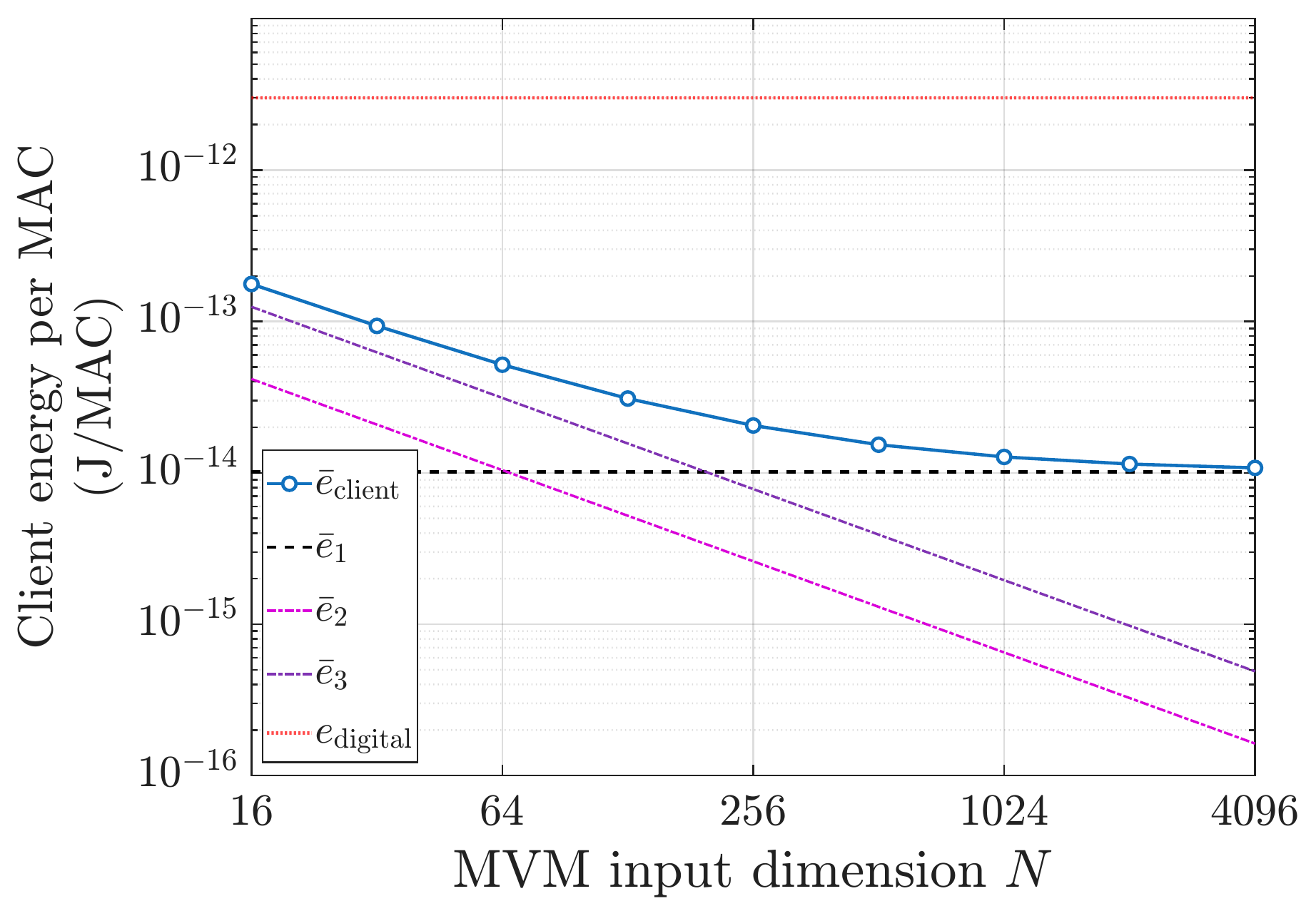}
\caption{Client energy per MAC versus MVM input dimension.}
\label{fig:mvm_scaling}
\end{minipage}
\end{figure*}

\subsection{Robustness to RF Impairments}\label{sec:sim_rf}
\begingroup 
\setlength{\abovedisplayskip}{3pt}
\setlength{\belowdisplayskip}{3pt}
\setlength{\abovedisplayshortskip}{0pt}
\setlength{\belowdisplayshortskip}{3pt}

Practical impairments of RF circuits can lead to residual errors at the receiver. To assess their effect, Fig.~\ref{fig:robustness_rf} plots the classification accuracy against the aggregate RF error-vector magnitude (EVM) $\xi$, with the nominal root NMSE $\epsilon=0.18$. Following the additive Gaussian residual-RF model in \cite{studer2010residual}, we represent the aggregate residual by a zero-mean Gaussian output-error term $\mathbf d_{\mathrm{RF},k}^{(\ell)}$ independent of thermal noise. Following Section~\ref{sec:Computing_Accuracy_Metric}, the output error is $\widehat{\mathbf y}_k^{(\ell)}-\mathbf y_k^{(\ell)}=\widetilde{\mathbf n}_k^{(\ell)}+\mathbf d_{\mathrm{RF},k}^{(\ell)}$. We define ${\xi^2=\mathbb E\|\mathbf d_{\mathrm{RF},k}^{(\ell)}\|_2^2/\mathbb E\|\mathbf y_k^{(\ell)}\|_2^2}$, using the error-to-reference power normalization of EVM in the 3GPP BS radio-transmission specifications \cite[Sec.~6.5.2]{etsi_bs}, with the MVM output $\mathbf{y}_k^{(\ell)}$ as the reference. In this case, we have
\begin{equation}
\mathrm{NMSE}_{k,\mathrm{tot}}^{(\ell)}=\epsilon^2+\xi^2,
\label{eq:impairment_nmse}
\end{equation}
in which $\mathrm{NMSE}_{k,\mathrm{tot}}^{(\ell)}$ is the total output NMSE considering both the thermal noise and the residual RF errors.

Fig.~\ref{fig:robustness_rf} shows that the classification accuracy decreases as the aggregate EVM increases, with a more pronounced reduction at larger errors. At a typical value of $\xi=0.15$, the classification accuracy remains \minorchange{$97.69\%$}. This corresponds to a relative reduction of only \minorchange{$0.76\%$} from the thermal-noise-only accuracy of \minorchange{$98.44\%$}. This shows the robustness of the proposed analog RF computing-based edge inference to the residual RF impairments. A detailed characterization of the robustness to more specific types of impairments is left for future work.
\par\endgroup

\subsection{Robustness to Reference Power Mismatch}\label{sec:sim_reference_power}

The reference output power $\bar S_{\mathrm{ref}}^{(\ell)}$ in Section~\ref{sec:Computing_Accuracy_Metric} is estimated offline and used to set the NMSE constraint. During deployment, the actual output power $\bar S_k^{(\ell)}$ may differ from this reference. We have $\mathrm{NMSE}_{k}^{(\ell)}=(\bar S_{\mathrm{ref}}^{(\ell)}/\bar S_k^{(\ell)})\mathrm{NMSE}_{k,\mathrm{ref}}^{(\ell)}$.

We consider two cases. When $\bar S_k^{(\ell)}\geq\bar S_{\mathrm{ref}}^{(\ell)}$, the realized NMSE is no greater than the reference NMSE, which preserves the computing accuracy requirement and potentially improves the inference accuracy. When $\bar S_k^{(\ell)}<\bar S_{\mathrm{ref}}^{(\ell)}$, the lower actual output power raises the realized NMSE, which may exceed the prescribed target and reduce the inference accuracy.

To protect against the latter case, we compare two designs. The 0-dB margin design uses the original NMSE target $(\epsilon_k^{(\ell)})^2$. The 3-dB margin design tightens this target by 3~dB, approximately halving it when setting the reference NMSE constraint. This safety margin allows the actual output power to fall to approximately half of the reference output power while maintaining the original NMSE requirement.

Fig.~\ref{fig:calibration_mismatch} shows the classification accuracy versus $10\log_{10}(\bar S_k^{(\ell)}/\bar S_{\mathrm{ref}}^{(\ell)})$ (dB). A value of $0$~dB on the x-axis corresponds to equal actual and reference output powers. A value of $-3$~dB corresponds to an actual output power approximately half of the reference output power. The accuracy decreases as the actual output power falls. At $-3$~dB, the 3-dB margin design maintains \minorchange{$98.44\%$} accuracy, a relative reduction of only \minorchange{$0.37\%$} from its accuracy at $\bar S_k^{(\ell)}=\bar S_{\mathrm{ref}}^{(\ell)}$, while the 0-dB margin design gives an accuracy of \minorchange{$97.24\%$}. These results show that a safety margin improves the robustness of analog RF computing-based edge inference to reference power mismatch.

\subsection{Scalability to Large-Scale NNs}\label{sec:scaling}
To examine the scalability and energy efficiency of analog RF computing for large-scale NN inference, Fig.~\ref{fig:mvm_scaling} plots the client-side energy per MAC $\bar e_{\mathrm{client}}$ and its three components against input dimension $N$. Larger input dimensions represent wider NN layers or layers processing more input features, i.e., larger NNs. We consider a representative MVM and omit the client and layer indices here. We evaluate complex $6\times N$ MVMs with unit-norm weight rows, unit-variance Gaussian distributed inputs, $|a|=0.5$, $\epsilon=0.18$, $L_{\mathrm{fft}}=8N$, and the nominal system and hardware parameters in Table~\ref{tab:constants}.

Fig.~\ref{fig:mvm_scaling} shows that the ADC sampling energy $\bar e_2$ and the digital decoding energy $\bar e_3$ per MAC decrease as $1/N$, since each readout serves more MACs. The waveform-generation energy $\bar e_1$ remains constant as shown in \eqref{eq:e1_beta}. Consequently, increasing the layer input dimension $N$ reduces $\bar e_{\mathrm{client}}$, which gradually saturates at the waveform-generation energy floor. The advantage in energy efficiency is more pronounced for large-scale NN inference. These results support the scalability of analog RF computing for large-scale NN inference.

\section{Conclusion}\label{sec:conclusion}
In this paper, we developed a physical layer design framework for analog RF computing-based edge AI inference in MU-MIMO wireless systems. We derived the computing accuracy and energy models that connect the NMSE and energy consumption of analog MVMs to the BS beamformer and client-side scaling coefficients. Based on these models, we formulated a joint BS-side beamforming and client-side scaling problem subject to computing accuracy, transmit power, and hardware operating-region constraints. We further derived the closed-form client-side scaling structure and channel-subspace structure of the BS beamformer, which led to a reduced-dimensional SCA algorithm. We showed how the proposed physical layer design supports both uniform- and mixed-precision inference. These results suggest that treating downlink beamforming as a computing resource, rather than only a communication resource, can enable a new paradigm of energy-efficient edge AI services over wireless networks. Future work will further investigate the integration of data communication and analog RF computing-based edge AI over shared wireless resources, building on the initial results in our follow-up work~\cite{yu2026integrated}. In addition, we will investigate more spectrum-efficient weight-encoded waveforms and input-encoded waveforms that can process a batch of inputs at the same time. Lastly, it is also interesting to extend the proposed paradigm to scenarios where the BS hosts a number of NN models and serves many rapidly moving clients, such as low-altitude wireless networks~\cite{zhang2026multimodal}.

\Needspace*{5\baselineskip}
\bibliographystyle{IEEEtran}
\bibliography{references_TWC_analog_RF_computing}

\end{document}